\documentclass{imsart}

\usepackage{natbib}
\usepackage{amsmath,bm,amsfonts,color,amsthm}
\usepackage{graphicx}
\usepackage{epstopdf}
\usepackage{caption}
\usepackage{subcaption}

\usepackage{float}
\restylefloat{table}
\usepackage{booktabs}

\newtheorem{lemma}{Lemma}
\newtheorem{theorem}{Theorem}
\newtheorem{definition}{Definition}
\usepackage{mathtools}	
\usepackage{algorithm}
\usepackage{algorithmic}
\usepackage[switch]{lineno}
\RequirePackage[colorlinks,citecolor=blue,urlcolor=blue]{hyperref}
\begin{document}
\begin{frontmatter}
\title{ Unique Entity Estimation with Application to the Syrian Conflict}
\runtitle{Entity Estimation with Application to the Syrian Conflict}
\begin{aug}
  \author{\fnms{Beidi} \snm{Chen}\ead[label=e1]{beidi.chen@rice.edu}},
  \author{\fnms{Anshumali} \snm{Shrivastava}\ead[label=e2]{anshumali@rice.edu}}, \and
    \author{\fnms{Rebecca C.} \snm{Steorts}\corref{}\ead[label=e3]{beka@stat.duke.edu}}
  \runauthor{Chen, Shrivastava, and Steorts}
    \address{Department of Computer Science \\
   Rice University, Houston, TX USA\\ \printead{e1,e2}}
      \address{Department of Statistical Science \\and Computer Science\\ Duke University \\
  Durham, NC USA\\ \printead{e3}}
\end{aug}

\begin{abstract}
%Our work is motivated to estimate the number of documented identifiable deaths in the ongoing Syrian conflict.
Entity resolution identifies and removes duplicate entities in large, noisy databases and has grown in both usage and new developments as a result of increased data availability. Nevertheless, entity resolution has tradeoffs regarding assumptions of the data generation process, error rates, and computational scalability that make it a difficult task for real applications. In this paper, we focus on a related problem of unique entity estimation, which is the task of estimating the unique number of entities and associated standard errors in a data set with duplicate entities. Unique entity estimation shares many fundamental challenges of entity resolution, namely, that the computational cost of all-to-all entity comparisons is intractable for large databases. To circumvent this computational barrier, we propose an efficient (near-linear time) estimation algorithm based on locality sensitive hashing. Our estimator, under realistic assumptions, is unbiased and has provably low variance compared to existing random sampling based approaches. In addition, we empirically show its superiority over the state-of-the-art estimators on three real applications. The motivation for our work is to derive an accurate estimate of the documented, identifiable deaths in the ongoing Syrian conflict. Our methodology, when applied to the Syrian data set, provides an estimate of $191,874 \pm 1772$ documented, identifiable deaths, which is very close to the Human Rights Data Analysis Group (HRDAG) estimate of 191,369. Our work provides an example of challenges and efforts involved in solving a real, noisy challenging problem where modeling assumptions may not hold.
\end{abstract}


\end{frontmatter}

\section{Introduction}
\label{sec:intro}

Our work is motivated by a real estimation problem associated with the ongoing conflict in Syria. While deaths are tremendously well documented, it is hard to know how many unique individuals have been killed from conflict-related violence in Syria. Since March 2011, increasing reports of deaths have appeared in both the national and international news. There are many inconsistencies from various media sources, which is inherent due to the data collection process and the fact that reported victims are documented by multiple sources. Thus, our ultimate goal is to determine an accurate number of documented,  identifiable deaths (with associated standard errors) because such information may contribute to future transitional justice and accountability measures. For instance, statistical estimates of death counts have been introduced as evidence in national court cases and international tribunals investigating the responsibility of state leaders for crimes against humanity \citep{hrdag-2017}.

%The task shows a classical example of noisy and unreliable information in big data, which often brings a new set of challenges for statistical estimation.

The main challenge with reliable death estimation of the Syrian data set is the fact that individuals who are documented as dead are often duplicated in the data sets. To address this challenge, one could employ entity resolution (de-duplication or record linkage), which refers to the task of removing duplicated records in noisy datasets that refer to the same entity  \citep{liseo_2011, sadinle_2014, getoor_2006, gutman_2013,  McCallumWellner04, fellegi_1969}. Entity resolution is fundamental in many large data processing applications. Informally, let us assume that each entity (records) is a vector in $\mathbb{R}^D$. Then given a data set of $M$ records aggregated from many data sources with possibly numerous duplicated entities perturbed by noise, the task of entity resolution is to identify and remove the duplicate entities. For a review of entity resolution see \citep{winkler_2006, christen_2012, liseo2013some}.

One important subtask of entity resolution is estimating the number of unique entities (records) $n$ out of $M > n$ duplicated entities, which we call \emph{unique entity estimation}. Entity resolution is a more difficult problem because it requires one to link each entity to its associated duplicate entities. To obtain high-accuracy entity resolution, the algorithms must at least evaluate a significant amount of pairs for potential duplicates to ensure a link is not missed. Due to this (and to the best of our knowledge), accurate entity resolution algorithms scale quadratically or higher ($> O(M^2))$ making them computationally intractable for large data sets.  Reducing the computational cost in entity resolution is known as blocking, which, via deterministic or probabilistic algorithms, places similar records into blocks or bins \citep{christen_2012, steorts14comparison}. The computational efficiency comes at the cost of missed links and reduced accuracy for entity resolution. Further, it is not clear if we can use these crude but cheap entity resolution sub-routines for unbiased estimation of unique entities with strong statistical guarantees.

%{\bf Our Focus:}
The primary focus of this paper is on developing a \emph{unique entity estimation} algorithm that is motivated by the ongoing conflict in Syria and has the following desiderata:
\begin{enumerate}
  \item The estimation cost should be significantly less than quadratic ($O(M^2)$). In particular, any methodology requiring one to evaluate all pairs for linkage is not suitable. This is crucial for the Syrian data set and other large, noisy data sets (Section~\ref{sec:SyriaChallenges}).
  \item To ensure accountability regarding estimating the unique number of documented identifiable victims in the Syrian conflict, it is essential to understand the statistical properties of any proposed estimator. Such a requirement eliminates many heuristics and rule-based entity resolution tasks, where the estimates may be very far from the true value.
  \item  In most real entity resolution tasks, duplicated data can occur with arbitrarily large changes including missing information, which we observe in the Syrian data set, and standard modeling assumptions may not hold due to the noise inherent in the data. Due to this, we prefer not to make strong modeling assumptions regarding the data generation process.
\end{enumerate}

\subsection{Related Work for Unique Entity Estimation}
\label{sec:RelatedWork}

The three aforementioned desiderata eliminate all but random sampling-based approaches. In this section, we review them briefly.

To our knowledge, only two random sampling based methodologies satisfy such requirements. \cite{1978paper} proposed sampling a large enough subgraph to estimate the total number of connected components based on the properties of the sub-sampled subgraph. Also, \cite{chazelle2005approximating} proposed finding connected components with high probability by sampling random vertices and then visiting their associated components using breadth-first search (BFS). One major issue with random sampling is that most sampled pairs are unlikely to be matches (no edge) providing nearly no information, as the underlying graph is generally very sparse in practice. Randomly sampling vertices and running BFS required by~\cite{chazelle2005approximating} are very likely to result in singleton vertices because many records are themselves unique in entity resolution data sets. In addition, finding all possible connections of a given vertex would require $O(M)$ query for edges. A query for edges corresponds  to the query for actual link between two records. Sub-sampling a sub-graph, as in~\cite{1978paper}, of size $s$ requires $O(s^2)$ edge queries to completely observe it. Thus, $s$ should be reasonably small in order to scale. Unfortunately, requiring a small $s$ hurts the variance of the estimator. We show that the accuracy of both aforementioned methodologies is similar to the non-adaptive variant of our estimator which has provably large variance. In addition, we show both theoretically and empirically that the methodologies based on random sampling lead to poor estimators.

While some methods have recently been proposed for accurate estimation of unique records, they belong to the Bayesian literature and have difficulty scaling due to the curse of dimensionality with Markov chain Monte Carlo \cite{steorts??bayesian, sadinle_2014, liseo_2011}. The evaluation of the likelihood itself is quadratic. Furthermore, they rely on a strong assumption about the specified generative models for the duplicate records. Given such computational challenges with the current state of the methods in the literature, we take a simple approach, especially given the large and constantly growing  data sets that we seek to analyze. We focus on practical methodologies that can easily scale to large data sets with minimal assumptions. Specifically, we propose a unique entity estimation algorithm with sub-quadratic cost, which can be reduced to approximating the number of connected components in a graph with sub-quadratic queries for edges (Section~\ref{form}).

The rest of the paper proceeds as follows. Section \ref{syriadata} provides our motivational application from the Syrian conflict
and Section \ref{sec:SyriaChallenges} remarks on the main challenges of the Syrian data set and our proposed methodology. Section \ref{sec:lsh} provides background on variants of locality sensitive hashing (LSH), which is essential to our proposed methodology.  Section \ref{proposal} provides our proposed methodology for unique entity estimation, which is the first formalism of using efficient adaptive LSH on edges to estimate the connected components with sub-quadratic computational time. (An example of our approach is given in section \ref{sec:toy}). More specifically, we draw connections between our methodology and
 random and adaptive sampling in section \ref{analysis}, where we show under realistic assumptions that our estimator is theoretically unbiased and has provably low variance. In addition, in section \ref{sec:randVsAdaptive}, we compare random and adaptive sampling for the Syrian data set, illustrating the strengths of adaptive sampling. In section \ref{sec:missingVariant}, we introduction the variant of LSH used in our paper. Section \ref{putit} provides our complete algorithm for unique entity estimation.
Section \ref{experiments} provides evaluations of all the related estimation methods on three real data sets from the music and food industries as well as official statistics. Section \ref{syria} reports the documented identifiable number of deaths in the Syrian conflict (with a standard error).

\subsection{The Syrian Conflict}
\label{syriadata}
Thanks to Human Rights Data Analysis Group (HRDAG), we have access to four databases from the Syrian conflict which cover roughly the same period, namely March 2011 -- April 2014, namely, the Violation Documentation Centre (VDC), Syrian Center for Statistics and Research (CSR-SY), Syrian Network for Human Rights (SNHR), and Syria Shuhada website (SS). Each database lists a different number of recorded victims killed in the Syrian conflict, along with available identifying information including full Arabic name, date of death, death location, and gender.\footnote{These databases include documented identifiable victims and not those who are missing in the conflict, hence, any estimate reported only refers to the data at hand.}

%This list does not contain individuals that are missing, and a victim is listed as killed if they are identifiable.

Since the above information is collected indirectly, such as through friends and religious leaders, or traditional media resources, it naturally comes with many challenges. The data set has biases, spelling errors, and missing values. In addition, it is well known that there are duplicate entities present in the data sets, making estimation more difficult.
%In addition, it is well known that many records refer to the same entity (individual), a problem known as entity resolution (record linkage or de-duplication).
The ambiguities in Arabic names make the situation significantly worse as there can be a large textual difference between the full and short names in Arabic. (It is not surprising that the Syrian data set has such biases given that the data is collected in the midst of a conflict).

Such ambiguities and lack of additional information make entity resolution on this data set considerably challenging~\citep{price2014updated}. Owing to the significance of the problem, HRDAG has provided labels for a large subset of the data set. More specifically, five different human experts from the HRDAG manually reviewed pairs of records in the four data sets, classifying them as matches if referred to the same entity and non-matches otherwise. {\em Our first {\bf goal} is to accurately estimate the number of unique victims.} Obtaining a match or non-match label of a given record pair may require momentous cost such as manual human supervision or involving sophisticated machine learning. Given that coming up with hand-matched data is a costly process,  {\em our second {\bf goal}} is to provide a proxy, automated mechanism to create labeled data. (More information regarding the Syrian data set can be found in Appendix \ref{sec:syrian}).

\subsection{Challenges and Proposed Solutions}
\label{sec:SyriaChallenges}
%Evaluating all pairs for match/mismatch seems to solve the problem.
Consider evaluating the Syrian data set using all-to-all records comparisons to remove duplicate entities. With approximately 354,000 records from the Syrian data set, we have around 63 billion pairs ($6.3 \times 10^{10}$). Therefore, it is impractical to classify all these pairs as matches/non-matches reliably. We cannot expect a few experts (five in our case) to manually label 63 billion pairs. A simple computation of all pairwise similarity (63 billion) takes more than 8 days on a heavyweight machine that can run 56 threads in parallel (28 cores in total). In general, this quadratic computational cost is widely considered infeasible for large data sets. Algorithmic labeling of every pair, even if possible for relatively small datasets, is neither reliable nor efficient. Furthermore, it is hard to understand the statistical properties of algorithmic labelling of pairs. Such challenges, therefore, motivate us to focus on the estimation algorithm with constraints mentioned in Section~\ref{sec:intro}.

{\bf Our Contributions:} We formalize unique entity estimation as approximating the number of connected components in a graph with sub-quadratic $\ll O(M^2)$ computational time. We then propose a generic methodology that provides an estimate in sample (with standard errors). Our proposal leverages locality sensitive hashing (LSH) in a novel way for the estimation process, with the required computational complexity that is less than quadratic. Our proposed estimator is unbiased and has provably low variance compared to random sampling based approaches. To the best of our knowledge this is the first use of LSH for unique entity estimation in an entity resolution setting.
Our unique entity estimation procedure is broadly applicable to many applications, and we illustrate this on three additional real, fully labelled, entity resolution data sets, which include the food industry, the music industry, and an application from official statistics. In the absence of ground truth information, we estimate that the number of documented identifiable deaths for the Syrian conflict is 191,874, with standard deviation of 1,772, reported casualties, which is very close to the 2014 HRDAG estimate of 191,369. This clearly demonstrates the power of our efficient estimator in practice, which does not rely on any strong modeling assumptions. Out of 63 billion possible pairs, our estimator only queries around 450,000 adaptively sampled pairs ($\simeq O(M)$) for labels, yielding a 99.99\% reduction. The labelling was done using support vector machines (SVMs) trained on a small number of hand-matched, labeled examples provided by five domain experts. Our work is an example of the efforts required to solve a real noisy challenging problem where modeling assumptions may not hold.

%Our contributions in this paper are four-fold. First, we propose a methodology that provides an estimate in sample (with standard errors) using adaptive locality sensitive hashing, which has computational complexity that is less than quadratic. Second, we propose an estimator for unique entity estimation that is unbiased and has provably low variance. Third, in the absence of ground truth information, our methodology suggests that it can provide a proxy for such labels. Finally, we illustrate the success of our framework on a case study from the Syrian conflict in addition to three other real data sets, making comparisons to the current state-of-the-art in the literature.

\section{Variants of Locality Sensitive Hashing (LSH)}
\label{sec:lsh}
%\textcolor{red}{Redone: AS please check. I have tried to make it clear, accurate and minimal!}
In this section, we first provide a review of LSH and min-wise hashing, which is crucial to our proposed methodology. We then introduce a variant of LSH --- Densified One Permutation Hashing (DOPH), which is essential to our proposed algorithm for unique entity estimation in terms of scalability.

\subsection{Shingling}
In entity resolution tasks, each record can be represented as a string of information. For example, each record in the Syrian data set can be represented as a short \textit{text} description of the person who died in the conflict. In this paper, we use a k-grams based shingle representation, which is the most common representation of text data and naturally gives a set token (or k-grams).
That is, each record is treated as a string and is replaced by a ``bag'' (or ``multi-set'') of length-$k$ contiguous sub-strings that it contains. Since we will use a k-gram based approach to transform the records, our representation of each record will also be a set, which consists of all the $k$-contiguous characters occurring in record string. As an illustration, for the record \text{BAKER, TED}, we separate it into a 2-gram representation. The resulting set is the following: $$\text{BA, AK, KE, ER, ER, TE, ED}.$$ In another example, consider \text{Sammy, Smith}, whose 2-gram set representation is $$\text{SA, AM, MM, MY, MS, SM, MI, IT, TH}.$$ We now have two records that have been transformed into a 2-gram representation. Thus, for every record (string) we obtain a set $\subset \mathcal{U}$, where the universe $\mathcal{U}$ is the set of all possible $k$-contiguous characters.
% k-grams based shingle representation is the most common representation of text data, which naturally gives a set token (or k-grams). Since we will use a k-gram based approach to transform the records, our representation of each record will also be a set. The set will consist of all the $k$-contiguous characters occurring in record string. As an illustration, for the record \text{BAKER, TED}, we separate it into a 2-gram representation. The resulting set is the following: $$\text{BA, AK, KE, ER, ER, TE, ED}.$$ In another example, consider \text{Sammy, Smith}, whose 2-gram set representation is $$\text{SA, AM, MM, MY, MS, SM, MI, IT, TH}.$$ We now have two records that have been transformed into a 2-gram representation. Thus, for every record (string) we obtain a set $\subset \mathcal{U}$, where the universe $\mathcal{U}$ is the set of all possible $k$-contiguous characters.

\subsection{Locality Sensitive Hashing}
LSH---a two-decade old probabilistic technique and method for dimension reduction---comes with sound mathematical formalism and guarantees. LSH is widely used in computer science and database engineering as a way of rapidly finding approximate nearest neighbors \citep{Proc:Indyk_STOC98,gionis_1999}. Specifically, the variant of LSH that we utilize is scalable to large databases, and allows for similarity based sampling of entities in less than a quadratic amount of time.

In LSH, a hash function is defined as $y = h(x),$ where $y$ is the \emph{hash code} and $h(\cdot)$ the \emph{hash function}. A \emph{hash table} is a data structure that is composed of \emph{buckets} (not to be confused with blocks), each of which is indexed by a \emph{hash code}. Each reference item $x$ is placed into a bucket $h(x).$

More precisely, LSH is a family of function that map vectors to a discrete set, namely, $h:\mathbb{R}^D \rightarrow \{1, \ 2,\cdots , M\}$, where $M$ is in finite range. Given this family of functions, similar points (entities) are likely to have the same hash value compared to dissimilar points (entities). The notion of similarity is specified by comparing two vectors of points (entities), $x$ and $y.$ We will denote a general notion of similarity by $\text{SIM}({x,y}).$ In this paper, we only require a relaxed version LSH, and we define this below. Formally, a LSH is defined by the following definition below:
%LSH are family of functions mapping vectors to a discrete set, i.e., $h:\mathbb{R}^D \rightarrow \{1, \ 2,\cdot , M\}$ where $M$ is in finite range. In this family, similar points are likely to have the same hash value compare to dissimilar points. The notion of similarity is specified.  The exact definition is complicated and involved; we give a relaxed version of the definition for the sake of clarity.

\begin{definition} (Locality Sensitive Hashing (LSH))\ Let $x_1, \ x_2, \ y_1, \ y_2 \in \mathbb{R}^D$ and suppose $h$ is chosen uniformly from a family $\mathcal{H}.$ Given a similarity metric, $\text{SIM}(x,y)$, $\mathcal{H}$ is locality sensitive if $\text{SIM}(x_1,x_2)\ge Sim(y_2,y_3)$ then ${Pr}_\mathcal{H}(h(x_1) = h(x_2)) \ge {Pr}_\mathcal{H}(h(y_1) = h(y_2)),$ where ${Pr}_\mathcal{H}$ is the probability over the uniform sampling of $h$.
%Under a given similarity metrix $\text{SIM}(x,y)$, a family $\mathcal{H}$ is locality sensitive if for any four points $x_1, \ x_2, \ y_1, \ y_2 \in \mathbb{R}^D$  and $h$ chosen uniformly from $\mathcal{H}$ satisfies the following condition: if $\text{SIM}(x_1,x_2)\ge Sim(y_2,y_3)$ then ${Pr}_\mathcal{H}(h(x_1) = h(x_2)) \ge {Pr}_\mathcal{H}(h(y_1) = h(y_2)),$ where ${Pr}_\mathcal{H}$ is the probability over the uniform sampling of $h$.
\end{definition}

The above definition is sufficient condition for a family of function to be LSH. While many popular LSH families satisfy the aforementioned property, we only require this condition for the work described herein. For a complete review of LSH, we refer to \cite{rajaraman_2012}.
% Many popular LSH families satisfy the above property. We only need this condition for this paper.
%More recently, locality sensitive hashing has been utilized as a form of blocking in entity resolution, where
%one tries to achieve scalability and avoid all-to-all record comparisons by partition records into ``partitions" or ``blocks" either using deterministic or probabilistic methods [[CITE Christen, Herzog]]. However, the theoretical understanding concerning the effect of LSH on the record linkage problem is missing to a large extent.

\subsection{Minhashing}
One of the most popular forms of LSH is minhashing~\citep{Proc:Broder}, which has two key properties --- a type of similarity and a type of dimension reduction. The type of similarity used is the Jaccard similarity and the type of dimension reduction is known as the minwise hash, which we now define.

Let $\{0,1\}^D$ denote the set of all binary $D$ dimensional vectors, while $\mathbb{R}^D$ refers to the set of all $D$ dimensional vectors (of records). Note that records can be represented as a binary vector (or set) representation via shingling, BoW, or combining these two methods. More specifically, given two record sets (or equivalently binary vectors) $x,y \in \{0,1\}^D,$ the Jaccard similarity between $x, y \in \{0,1\}^D$  is\begin{align*}
\mathcal{J} = \frac{|x \cap y|}{|x \cup y|},
%\footnote{Minwise hashing is the popular locality sensitive hashing (LSH) for Jaccard similarity.}.
\end{align*}
where $|\cdot|$ is the cardinality of the set.

More specifically, the minwise hashing family applies a random permutation $\pi$, on the given set $S$, and stores only the minimum value after the permutation mapping, known as the \emph{minhash}.  Formally, the minhash is defined as $h_{\pi}^{min}(S) = \min(\pi(S))$, where $h(\cdot)$ is a hash function.

Given two sets $S_1$ and $S_2$, it can be shown by an elementary probability argument that
\begin{equation}
\label{eq:MinHash}
Pr_{\pi}({h_{\pi}^{min}(S_1) = h_{\pi}^{min}(S_2)) =  \frac{|S_1 \cap S_2|}{| S_1 \cup S_2|}},
% \mathcal{J}=\frac{a}{f_1+f_2 -a},
\end{equation}
where the probability is over uniform sampling of $\pi$. It follows from Equation~\ref{eq:MinHash} that minhashing is a LSH family for the Jaccard similarity.

\textbf{Remark}: In this paper, we utilize a shingling based approach, and thus, our representation of each record is likely to be very sparse. Moreover, \cite{shrivastava2014defense} showed that minhashing based approaches are superior compared to random projection based approaches for very sparse datasets.

\subsubsection{Densified One Permutation Hashing (DOPH)}
For realistically sized entity resolution tasks, sampling based on LSH requires hundreds of hashes (Section~\ref{sec:KLParametrizedLSH}). It is well known that computing several minwise hashes of data is a very costly operation~\citep{li2012gpu}. Fortunately, recent literature on DOPH has shown that it is possible to compute several hundreds or even thousands of hashes of the data vector in one pass with nearly identical statistical properties as minwise hashes \citep{shrivastava2014densifying,shrivastava2014improved,shrivastava2017optimal}. In this paper, we will use the most recent variant of DOPH, which is significantly faster in practice compared to minwise hashing ~\citep{shrivastava2017optimal}. Throughout the paper, our mention of minwise hashing will refer to the DOPH algorithm for computing minhashes, which we have just mentioned. The complete details can be found in the aforementioned papers.

\section{Unique Entity Estimation}
\label{proposal}
In this section, we provide notation used throughout the rest of the paper and provide an illustrative example.
We then propose our estimator, which is unbiased and has provably low variance. In addition, random sampling is a special case of our procedure as explained in section~\ref{sec:randVsAdaptive}.  Finally, we present our unique entity estimation algorithm in section~\ref{analysis}.
%In this section, we provide notation and an illustrative example for unique entity estimation. Next, we give our unbiased and
\subsection{Notation}
\label{form}
The problem of unique entity estimation can be reduced to approximating the number of connected components in a corresponding graph.
Given a data set with size $M$, we denote the records as
$$R = \{R_i| 1\leq i \leq M, \; i \in \mathbb{Z}\}.$$
%and
%define $Q(R_i, R_j) = 1$ if $R_i, R_j$ refer to the same entity and 0 otherwise. Now
Next, we define
$$Q(R_i, R_j) =\begin{cases}
1, & \text{if $R_i, R_j$ refer to the same entity }.\\
0, & \text{otherwise}. \end{cases}.$$
Let us represent the data set by a graph $G^* = (E, V),$ with vertices $E,V.$ Let vertex $V_i$ correspond to record $R_i$ and vertex $V_j$ correspond to record $R_j$. Then let edge $E_{ij}$ represent the linkage between records of $R_i$ and $R_j$ (or
vertex $V_i$ and $V_j$). More specifically, we can represent this by the following relationship:
\begin{align*}
V &= \{R_i|1\leq i \leq M, i \in \mathbb{Z}\}, \; \text{and} \;
E = \{(R_i, R_j)| \forall\; 1\leq i, j \leq M, \; Q(R_i, R_j) = 1\}.
\end{align*}

\subsection{Illustrative Example}
\label{sec:toy}

In this section, we provide an illustrative example of how six records are mapped to a graph $G^*$.
Consider record 3 (John) and record 5 (Johnathan) which correspond to the same entity (John Schaech). In $G^*$, there is an edge $E_{35}$ that connect these records, denoted by $V_3$ and $V_5.$
Now consider records 2, 4, and 6, which all refer to the same entity (Nicholas Cage). In $G^*$, there are edges $E_{24}, E_{26},$ and $E_{46}$  that connect these records, denoted by $V_2, V_4,$ and $V_6.$
Observe that each connected component in $G^*$ is a unique entity and also a clique. Therefore, our task is reduced to estimating the number of connected components in $G^*$.

\begin{figure}[ht]
	\begin{center}
		\includegraphics[width=\linewidth]{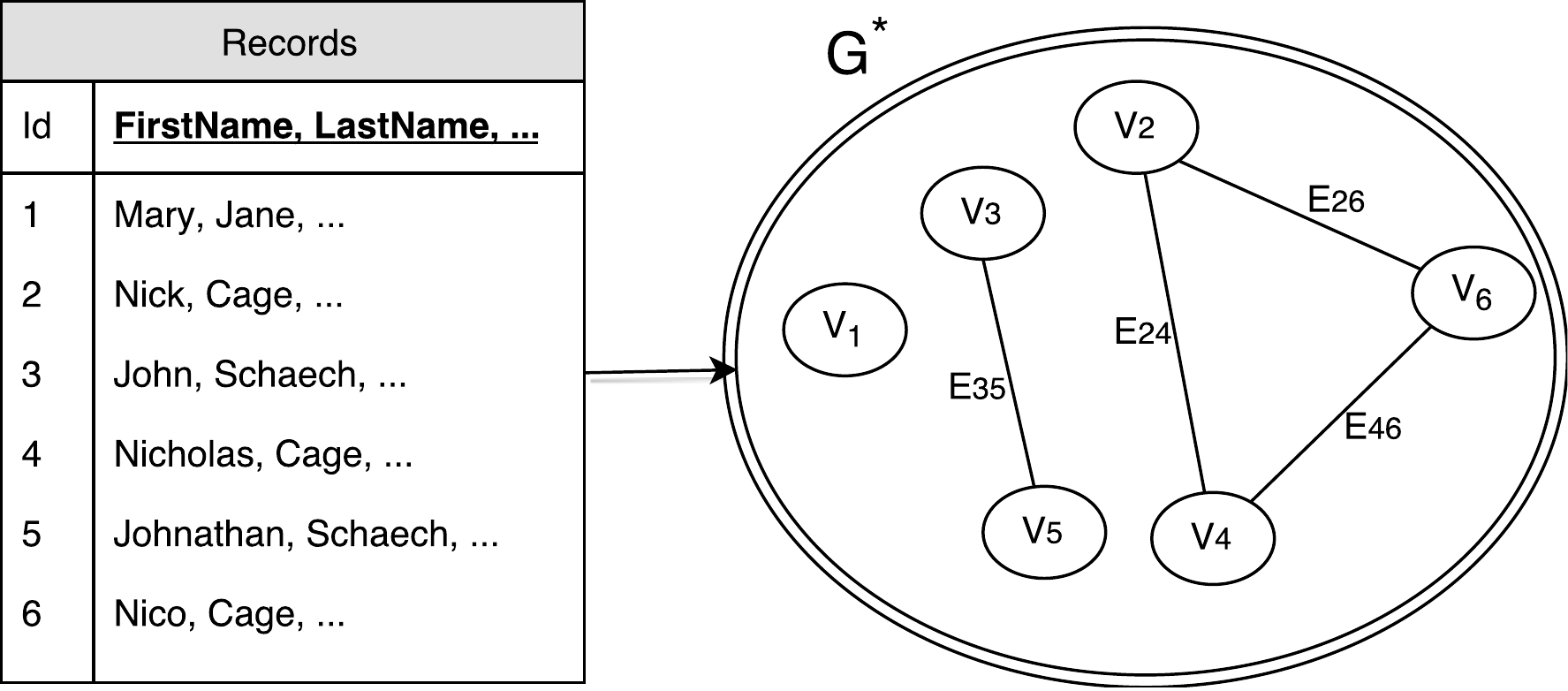}
		\caption{A toy example of mapping records to a graph, where vertices represent records and edges refer to the relation between records.}
		\label{map}
	\end{center}
	\vskip -0.2in
\end{figure}

\subsection{Proposed Unique Entity Estimator}
\label{analysis}
In this section, we propose our unique entity estimator and provide assumptions that are necessary for our estimation procedure to be practical (scalable).

Since we do not observe the edges of $G^*$ (the linkage), inferring whether there is an edge between two nodes (or whether two records are linked) can be costly, i.e., $O(M^2)$. Hence, one is constrained to probe a small set $\mathcal{S} \subset V \times V$ with $|\mathcal{S}|\ll O(M^2)$ of pairs and query if they have edges. The aim is to use the information about $\mathcal{S}$ to estimate the total number of connected components accurately. More precisely,  given the partial graph $G' = \{V,E'\}$, where $E' = E \cap \mathcal{S}$, one wishes to estimate the connected components $n$ of $G^* = \{V,E\}.$

One key property of our estimation process is that we do not make any modeling assumptions of how duplicate records are generated, and it is not immediately clear how we can obtain unbiased estimation. For sake of simplicity, we first assume the existence of an efficient (sub-quadratic) process that samples a small set (near-linear size) of edges $\mathcal{S}$, such that every edge in the original graph $G^*$ has (reasonably high) probability $p$ of being in $\mathcal{S}$. Thus, set $\mathcal{S}$, even though small, contains $p$ fraction of the actual edges. For sparse graphs, as in the case of duplicate records, such a sampler will be far more efficient than random sampling. Based on this assumption, we will first describe our estimator and its properties. We then show why our assumption about existence of adaptive sampler is practical by providing an efficient sampling process based on LSH (Section~\ref{samplep}).

\textbf{Remark}: It is not difficult to see that random sampling is a special case when $p = \frac{|\mathcal{S}|}{O(M^2)}$ which, as we show later, is a very small number for any accurate estimation.

Our proposed estimator and corresponding algorithm obtains the set of vertex pairs (or edges) $\mathcal{S}$  through an efficient (adaptive) sampling process and queries whether there is an edge (linkage) between each pair in $\mathcal{S}$.  Respectively, after the ground truth querying, we observe a sub-sampled graph $G'$, consisting of vertices returned by the sampler. Let $n_i'$ be the number of connected component of size $i$ in the observed graph $G'$, i.e., $n_1'$ is the number of singleton vertices, $n_2'$ is the number of isolated edges, etc. in $G'$. It is worth noting that every connected component in $G'$ is a part of some clique (maybe larger) in $G^*$.  Let $n_i^*$ denote the number of connected components (clique) of size $i$ in the original (unobserved) graph $G^*$.

Observe that under the sampling process, any original connected component, say $C_i^*$ (clique), will be sub-sampled and can appear as some possibly smaller connected component in $G'$. For example, a singleton set in $G^*$ will remain the same in $G'$. An isolated edge, on the other hand, can appear as an edge in $G'$ with probability $p$ and as two singleton vertices in $G'$ with probability $1-p$. A triangle can decompose into three possibilities with probability shown in Figure~\ref{p}. Each of these possibilities provides a linear equation connecting $n_i^*$ to $n_i'$. These equations up to cliques of size three are

\begin{align}
\label{n3}
&\mathbb{E}[n_3']= n_3^* \cdot p^2 \cdot (3-2p) \\
&\mathbb{E}[n_2'] = n_2^* \cdot p + n_3^*\cdot (3\cdot(1-p)^2\cdot p) \\
\label{n1}
&\mathbb{E}[n_1'] = n_1^* + n_2^* \cdot(2\cdot(1-p)) + n_3^* \cdot (3\cdot (1-p)^2).
\end{align}

Since we observe $n_i'$,  we can solve for the estimator of each $n_i^*$ and compute the number of connected components by summing up all $n_i^*$.

\begin{figure}[ht]
	\begin{center}
		\includegraphics[width=\linewidth]{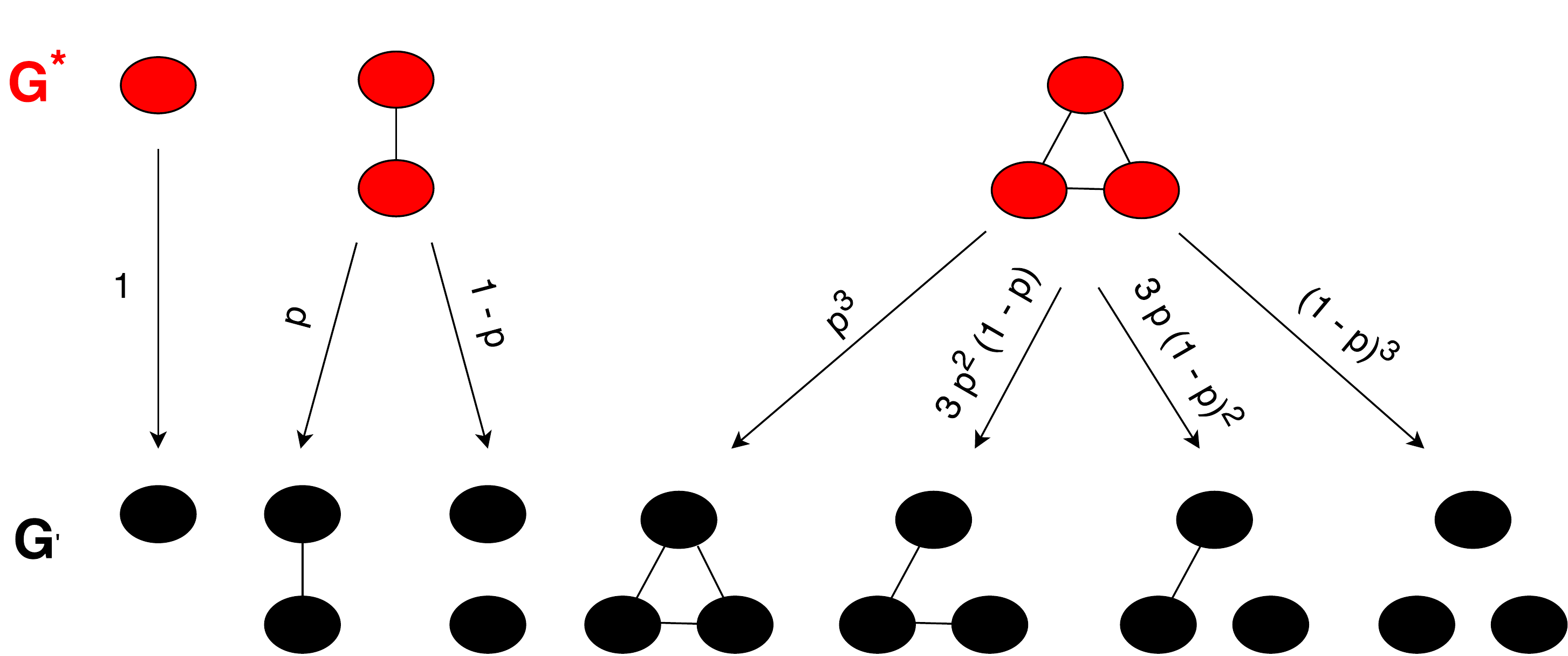}
		\caption{A general example illustrating the transformation and probabilities of connected components from $G^*$ to $G'.$}
		\label{p}
	\end{center}
	\vskip -0.2in
\end{figure}

Unfortunately, this process quickly becomes combinatorial, and in fact, is at least $\#P$ hard~\citep{doi:10.1137/0212053} to compute for cliques of larger sizes. A large clique of size $k$ can appear as many separate connected components and the possibilities of smaller size components it can break into are exponential ~\citep{wiki:xxx}.
Fortunately, we can safely ignore large connected components without significant loss in estimation for two reasons.
First, in practical entity resolution tasks, when $M$ is large and contains at least one string-valued feature, it is observed that \emph{most} entities are replicated no more than three or four times. Second,  a large clique can only induce large errors if it is broken into many connected components due to undersampling. According to \cite{erdos1960evolution}, it will almost surely stay connected if $p$ is high, which is the case with our sampling method.

{\bf Assumption:} As argued above, we safely assume that the cliques of sizes equal to or larger than 4 in the original graph would retain their structures, i.e., $\forall i \ge 4, \ n_i^* = n_i'$. With this assumption, we can write down the formula for estimating $n_1^*$, $n_2^*$, $n_3^*$ by solving Equations~\ref{n3}--\ref{n1} as,
\begin{align}
&n_3^*= \frac{\mathbb{E}[n_3']}{p^2 \cdot (3-2p) }, \quad n_2^* = \frac{\mathbb{E}[n_2'] - n_3^*\cdot (3\cdot(1-p)^2\cdot p)}{p}  \\
&n_1^* = \mathbb{E}[n_1'] - n_2^* \cdot(2\cdot(1-p)) - n_3^* \cdot (3\cdot (1-p)^2)
\end{align}
It directly follows that our estimator, which we call the Locality Sensitive Hashing Estimator (LSHE) for the number of connected components is given by
\begin{align}
\label{main}
\text{LSHE} &= n_1' +  n_2' \cdot \frac{2p-1}{p}  + n_3' \cdot \frac{1-6 \cdot (1-p)^2 \cdot p}{p^2 \cdot (3-2p) }
+ \sum_{i=4}^{M}n_i'.
\end{align}

\subsection{Optimality Properties of LSHE}
%\subsection{Theoretical Analysis}
We now prove two properties of our unique entity estimator, namely, that it is unbiased and that is has provably lower variance than random sampling approaches.

\begin{theorem}
	\label{thm:thm1}
	Assuming  $\forall i \ge 4, \ n_i^* = n_i'$, we have
	\begin{align}\label{eq:thm1}
	\mathbb{E}[\text{LSHE}] &=  n   \ \ \ \mbox{unbiased} \\
	\label{var}
	\mathbb{V}[\text{LSHE}] &= n_3^* \cdot \frac{(p-1)^2 \cdot (3p^2-p+1)}{p^2 \cdot (3-2p)} + n_2^*\frac{(1-p)}{p}
	\end{align}
	The above estimator is unbiased and the variance is given by Equation \ref{var}.
\end{theorem}
Theorem \ref{thm:thm2} proves the variance of our estimator is monotonically decreasing with $p$.
\begin{theorem}
	\label{thm:thm2}
	$\mathbb{V}[\text{LSHE}]$ is monotonically decreasing when $p$ increases in  range $(0,1]$.
\end{theorem}

The proof of Theorem \ref{thm:thm2} directly follows from the following Lemma \ref{lemma}.
\begin{lemma}
	\label{lemma}
	First order derivative of $\mathbb{V}[\text{LSHE}]$ is negative when $p \in (0,1]$.
\end{lemma}
Note that when $p=1$, $\mathbb{V}[\text{LSHE}] = 0$ which means the observed graph $G'$ is exactly the same as $G^*$. For detailed proofs of unbiasedness and Lemma \ref{lemma}, see Appendix \ref{sec:app}.

\subsection{Adaptive Sampling versus Random Sampling}
\label{sec:randVsAdaptive}
Before we describe our adaptive sampler, we briefly quantify the advantages of an adaptive sampling over random sampling for the Syrian data set by computing the differences between their variances. Let $p$ be the probability that an edge (correct match) is sampled. On the Syrian data set, our proposed sampler, described in next section, empirically achieves $p = 0.83$, by reporting around 450,000 sampled pairs ($O(M)$) out of the 63 billion possibilities ($O(M^2)$). Substituting this value of $p$, the corresponding variance can be calculated from Equation~\ref{var} as $$n_3^* \cdot 0.07 + n_2^*\cdot 0.204.$$

Turning to plain random sampling of edges, in order to achieve the same sample size above leads to $p$ as low as $\frac{4.5 \times 10^5}{6.3 \times 10^{10}} \simeq 6.9 \times 10^{-6}$. With such minuscule $p$, the resulting variance is $$n_3^* \cdot 6954620166 + n_2^*\cdot 144443.$$ Thus, the variance for random sampling is roughly $7 \times 10^{5}$ times the number of duplicates in the data set and $1 \times 10^{11}$ the number of triplets in the data set.

In section~\ref{experiments}, we illustrate that two other random sampling based algorithms of~\cite{chazelle2005approximating} and ~\cite{1978paper} also have poor accuracy compared to our proposed estimator. The poor performance of random sampling is not surprising from a theoretical perspective, and illustrates a major weakness empirically for the task of unique entity estimation with sparse graphs, where adaptive sampling is significantly advantageous.

\subsection{The Missing Ingredient: (K,L)-LSH Algorithm}
\label{sec:missingVariant}
%\subsection{The Missing Ingredient: Sub-quadratic Adaptive Sampling}
%\textcolor{red}{THIS SECTION IS TOO LONG AND THE SECTION TITLE IS CONFUSING. COME BACK TO THIS.}
Our proposed methodology, for unique entity estimation, assumes that we have an efficient algorithm that adaptively samples a set of record pairs, in sub-quadratic time. In this section, we argue that using a variant of LSH (Section~\ref{sec:lsh}) we can construct such an efficient sampler.

As already noted, we do not make any modeling assumptions on the generation process of the duplicate records. Also, we cannot assume that there is a fixed similarity threshold, because in real datasets duplicates can have arbitrarily large similarity. Instead, we rely on the observation that record pairs with high similarity have a higher chance of being duplicate records.  That is, we assume that when two  entities $R_i$ and $R_j$ are similar in their attributes, it is more likely that they refer to the same entities~\citep{christen_2012}.\footnote{The similarity metric that we use to compare sets of record strings is the Jaccard similarity.} We note that this probabilistic observation is the weakest possible assumption, and almost always true for entity resolution tasks because linking records by a similarity score is one simple way of approaching entity resolution \citep{christen_2012, winkler_2006, fellegi_1969}.

The similarity between entities (records) naturally gives us a notion of adaptiveness. One simple adaptive approach is to sample records pairs with probability proportional to their similarity. However, as a prerequisite for such sampling, we must compute all the pairwise similarities and associated probability values with every edge. Computing such a pairwise similarity score is a quadratic operation ($O(M^2)$) and is intractable for large datasets. Fortunately, recent work has shown that~\citep{spring2017scalable,spring2017new,luo2017Arrays} it is possible to sample pairs adaptively in proportion to the similarity in provably sub-quadratic time using LSH, which we describe in the next section.

%\subsubsection{(K,L)-parameterized Pair Sampling Algorithm}
\subsubsection{(K,L)-LSH Algorithm and Sub-quadratic Adaptive Sampling}
\label{sec:KLParametrizedLSH}

We leverage a very recent observation associated with the traditional $(K, L)$ parameterized LSH algorithm. The $(K, L)$ parameterized LSH algorithm is a popular similarity search algorithm, which given a query $q$, retrieves element $x$ from a preprocessed data set in sub-linear time ($O(KL) \ll M$) with probability $1-(1-\mathcal{J}(q,\ x)^K)^L$.  Here, $J$ denotes the Jaccard similarity between the query and the retrieved data vector $x$. Our proposed method leverages this $(K, L)$-parameterized LSH Algorithm, and we briefly describe the algorithm in this section. For complete details refer to~\cite{Report:E2LSH}.

Before we proceed, we define hash maps and keys. We use hash maps, where every integer (or key) is associated with a bucket (or a list) of records. In a hash map, searching for the bucket corresponding to a key is a constant time operation. Please refer to algorithms literature~\citep{rajaraman_2012} for details on hashing and its computational complexity. Our algorithm will require several hash maps, $L$ of them, where a record $R_i$ is associated with a unique bucket in every hash map. The key corresponding to this bucket is determined by minwise hashes of the record $R_i$. We encourage readers to refer to~\cite{Report:E2LSH} for implementation details.

More precisely, let $h_{ij}$, $i = \{1, \ 2,...,\ L\}$ and $j = \{1, \ 2,...,\ K\}$ be $K \times L$ minwise hash functions (Equation~\ref{eq:MinHash}) with each minwise hash function formed by independently choosing the underlying permutation $\pi$. Next, we construct $L$ meta-hash functions (or the keys) $H_i = \{h_{i,1}, h_{i,2},...,h_{i,K}\}$, where each of the $H_i$'s is formed by combining $K$ different minwise hash functions. For this variant of the algorithm, we need a total of $K \times L$ functions. With such $L$ meta-hash functions, the algorithm has two main phases, namely the data pre-processing and the sampling pairs phases, which we outline below.

\begin{itemize}
  \item {\bf Data Preprocessing Phase:} We create $L$ different hash maps (or hash tables), where every hash values maps to a bucket of elements.  For every record $R_i$ in the dataset, we insert $R_j$ in the bucket associated  with the key $H_i(R_j)$, in hash map $i = \{1, \ 2, ..., \ L\}$. To assign $K$-tuples $H_i$ (meta-hash) to a number in a fixed range, we use some universal random mapping function to the desired address range. See~\cite{Report:E2LSH, 2017arXiv170901190W} for details.
  \item {\bf Sample Pair Reporting:} \textcolor{black}{For every record $R_j$ in the dataset and from each table $i$, we obtain all the elements in the bucket associated with key $H_i(R_j)$, where $i = \{1,\ 2,...,\ L\}$. We then take the union of the $L$ buckets obtained from the $L$ hash tables, and denote this (aggregated) set by $A.$ We finally, report pairs of records $(R_i, R_j)$, where $R \in A$.}
\end{itemize}

\begin{theorem}
The \textcolor{black}{(K,L)-LSH Algorithm} reports a pair $(R_i,\ R_j)$ with probability $1 - (1-\mathcal{J}(R_i,\ R_j)^K)^L$, where $\mathcal{J}(R_i,\ R_j)$ is the Jaccard Similarity between record pairs $(R_i,\ R_j).$
\end{theorem}
{\bf Proof:} Since all the minwise hashes are independent due to an independent sampling of permutations, the probability that both $R_i$ and $R_j$ belong to the same bucket in any hash table $i$ is $\mathcal{J}(R_i,\ R_j)^K$. Note from equation~\ref{eq:MinHash}, each meta-hash agreement has probability $\mathcal{J}(R_i,\ R_j)$. Therefore, the probability that pair $(R_i, \ R_j)$ is missed by all the $L$ tables is precisely $(1-\mathcal{J}(R_i,\ R_j)^K)^L$, and thus, the required probability of successful retreival is the complement.

The probabilistic expression $1 - (1-\mathcal{J}(R_i,\ R_j)^K)^L$ is a monotonic function of the underlying similarity $Sim(q,y)$ associated with the LSH. In particular, higher similarity pairs have more chance of being retrieved. Thus, LSH provides the required sampling that is adaptive in similarity and is sub-quadratic in running time.

\subsubsection{Computational Complexity}
\label{sec:CC}
The computational compexity for sampling with $M$ records is $O(MKL)$. The procedure requires computing $KL$ minwise hashes for each record. This step is followed by adding every record to $L$ hash tables.  Finally, for each record, we aggregate $L$ buckets to form sample pairs. The result of monotonicity and adaptivity of the samples applies to any value of $K$ and $L$. We choose $O(K \times L) \ll O(M)$ such that we are able to get samples in sub-quadratic time. We further tune $K$ and $L$ using cross-validation to limit the size of our samples. In section~\ref{sec:parameterstudy}, we evaluate the effect of varying $K$ and $L$ in terms of the recall and reduction ratio. (For a review of the recall and reduction ratio, we refer to \cite{christen_2012}.) We address the precision at the very end of our experimental procedure to ensure that the recall, reduction ratio, and precision of our proposed unique entity estimation procedure are all as close to 1 as possible while ensuring that the entire algorithm is computationally efficient. For example, on the Syrian data set, we can generate 450,000 samples in less than $127 \sec$ with an adaptive sampling probability (recall) $p$ as high as $0.83$.  On the other hand, computing all pairwise similarities (63 billion) takes more than 8 days on the same machine with 28 cores capable of running 56 threads in parallel.

Next, we describe how this LSH sampler is related to the adaptive sampler described earlier in  Section~\ref{analysis}.

\subsubsection{Underlying Assumptions and Connections with $p$}
Recall that we can efficiently sample record pairs $R_i, \ R_j$ with probability $1-(1-J(R_i,\ R_j)^K)^L.$
%$1-(1-Sim(R_i,\ R_j)^K)^L$, where $Sim(R_i, \ R_j) = \mathcal{J}(R_i,R_j)$ is the Jaccard Similarity
%(Equation~\ref{eq:jacc}).
Since we are not making any modeling assumptions, we cannot directly link this probability to $p$, the probability of sampling the right duplicated pair (or linked entities) as required by our estimator $\text{LSHE}$. In the absence of any knowledge,
we can get the estimate of $p$ using a small set of labeled linked pairs $\mathcal{L}$. Specifically, we
%Given a small sample set of manually labeled $\mathcal{L}$ record pairs,
%which is often available for entity resolution datasets (and happens to be available for the Syrian data set),
we can estimate the value of $p$ by counting the fraction of matched pairs (true edges) from $\mathcal{L}$ reported by the sampling process.

Note that in practice there is no similarity threshold $\theta$ that guarantees that two record pairs are duplicate records. That is, it is difficult in practice to know a fixed $\theta$ where $\mathcal{J}(R_i,\ R_j) \ge \theta$ ensures that $R_i$ and $R_j$ are the same entities. However, the weakest possible and reasonable assumption is that high similarity pairs (textual similarity of records) should have higher chances of being duplicate records than lower similarity pairs.

Formally, this assumption implies that there exists a monotonic function $f$ of similarity $\mathcal{J}(R_i,\ R_j)$ such that the probability of any $R_i, \ R_j$ being a duplicate record is given by $f(\mathcal{J}(R_i,\ R_j))$.  Since our sampling probability $1-(1-\mathcal{J}(R_i,\ R_j)^K)^L$ is also a monotonic function of $\mathcal{J}(R_i,\ R_j)$, we can also write $$f(\mathcal{J}(R_i,\ R_j)) = g(1-(1-\mathcal{J}(R_i,\ R_j)^K)^L),$$ where $g$ is $f$ composed with $h^{-1}$ which is the inverse of $h(x) = 1-(1-x^K)^L$. Unfortunately, we do not know the form of $f$ or $g$.

Instead of deriving $g$ (or $f$), which requires additional implicit assumptions on the form of the functions, our process estimates $p$ directly. In particular, the estimated value of $p$ is a data dependent mean-field approximation of $g$, or  rather, $$p = \mathbb{E}[g(1-(1-\mathcal{J}(R_i,\ R_j)^K)^L)].$$ Crucially, our estimation procedure does not require any modeling assumptions regarding the generation process of the duplicate records, which is significant for noisy data sets, where such assumptions typically break.

%We stress on the fact that this does not require any modeling assumption about the duplicated entities, which is critical for real and significantly noisy datasets where all assumptions usually break.

\subsubsection{Why LSH?}
\label{sec:WHYLSH}
Although there are several rule-based blocking methodologies, LSH is the only one that is also a random adaptive sampler. In particular, consider a rule-based blocking mechanism, for example on the Syrian data set, which might block on the date of death feature. Such blocking could be a very reasonable strategy for finding candidate pairs. Note that it is still very likely that duplicate records can have different dates of death because the information could be different or misrepresented.  In addition, such a blocking method is deterministic, and different independent runs of the blocking algorithm will report the same set of pairs. Even if we find reasonable candidates, we cannot up-sample the linked records to get an unbiased estimate. There will be a systematic bias in the estimates, which does not have any reasonable correction. In fact, random sampling to our knowledge is the only known choice in the existing literature for an unbiased estimation procedure; however, as already mentioned, random uninformative sampling is likely to be very inaccurate.

LSH, on the other hand, can also be used as a blocking mechanism~\citep{steorts14comparison}. It is, however, more than just a blocking scheme; it is a provably adaptive sampler.
%a provable adaptive sampler.
Due to randomness in the blocking, different runs of sampler lead to different candidates, unlike deterministic blocking. We can also average over multiple runs to even increase the concentration of our estimates. The adaptive sampling view of LSH has come to light very recently~\citep{spring2017scalable,spring2017new,luo2017Arrays}. With adaptive sampling, we get much sharper unbiased estimators than the random sampling approach. To our knowledge, this is the first study of LSH sampling for unique entity estimation.

\subsection{Putting it all Together: Scalable Unique Entity Estimation}
\label{putit}

We now describe our scalable unique entity estimation algorithm.  As mentioned earlier, assume that we have a data set that contains a text representation of the $M$ records. Suppose that we have a reasonably sized, manually labeled training set $\mathcal{T}$. We will denote the set of sampled pairs of records given by our sampling process as $\mathcal{S}$. Note, each element of $\mathcal{S}$ is a pair. Then our scalable entity resolution algorithm consists of three main steps, with the total computational complexity $O(ML+ KL +|\mathcal{S}| +|\mathcal{T}|)$.  In our case, we will always have $|\mathcal{S}| \ll O(M^2)$ and $KL \ll M$ (in fact, $L$ will be a small constant), which ensures that the total cost is strictly sub-quadratic. The complete procedure is summarized in Algorithm~\ref{Est}.

\begin{enumerate}
    \item {\bf Adaptively Sample Record Pairs ($O(ML)$):} We regard each record $R_i$ as a short string and replace it by an ``n-grams" based representation. Then one computes $K \times L$  minwise hashes of each corresponding string. This can be done in a computationally efficient manner using the DOPH algorithm \citep{shrivastava2017optimal}, which is done in data reading time. Next, once these hashes are obtained, one applies the sampling algorithm described in section~\ref{samplep} in order to generate a large enough sample set, which we denote by $\mathcal{S}$. For each record, the sampling step requires exactly $L$ hash table queries, which are themselves $O(1)$ memory lookups. Therefore, the computational complexity of this step is $O(ML +KL).$
    \item {\bf Query each Sample Pairs:} Given the set of sampled pairs of records $\mathcal{S}$ from Step 1, for every pair of records in $\mathcal{S}$, we query whether these record pairs are a match or non-match. This step requires, $O(|S|)$, queries for the true labels. Here, one can use manually labeled data if it exists. In the absence of manually labeled data, we can also use a supervised algorithm, such as support vector machines or random forests, that is trained on the manually labeled set $\mathcal{T}$ (Section~\ref{syria}).
    \begin{enumerate}
        \item {\bf Estimate $p$:} Given the sampled set of record pairs $\mathcal{S},$ we need to know the value of $p$, the probability that any given correct pair is sampled. To do so, we use the fraction of true pairs sampled from the labeled training set $\mathcal{T}.$ The sampling probability $p$ can be estimated by computing the fraction of the matched pairs of training set records $\mathcal{T}_{match}$ appearing in $\mathcal{S}$. That is, we estimate $p$ (unbiasedly) by $$p = \frac{|\mathcal{T}_{match} \cap \mathcal{S}|}{|\mathcal{T}_{match}|}.$$ If $T$ is stored in a dictionary, then this step can be done on the fly while generating samples. It only costs $O(\mathcal{T})$ extra work to create the dictionary.

    \item {\bf Count Different Connected Components in $G'$ ($O(M+|\mathcal{S}|)$):}
        The resulting matched sampled pairs, after querying every sample for actual (or inferred) labels, form the edges of $G'$. We now have complete information about our sampled graph $G'$. We can now traverse $G'$ and count all sizes of connected components in $G'$ to obtain $n_1'$, $n_2'$, $n_3'$ and so on. Traversing the graph has computational complexity $O(M + |\mathcal{S}|)$ time using Breadth First Search (BFS).
    \end{enumerate}
    \item {\bf Estimate the Number of Connected Components in $G^*$ ($O(1)$):} Given the values of $p$, $n_1'$, $n_2'$, and $n_3'$  we use equation~\ref{main} to compute the unique entity estimator $\text{LSHE}$.
\end{enumerate}

\begin{figure}
%\begin{wrapfigure}{R}{0.47\textwidth}
	\begin{minipage}{\textwidth}
		\begin{algorithm}[H]
			\caption{LSH-Based Unique Entity Estimation Algorithm}
			\label{Est}
			\begin{algorithmic}[1]
				\STATE {\bfseries Input:} Records $R$, Labeled Set $\mathcal{T}$, Sample Size $m$
				\STATE {\bfseries Output:} $LSHE$
				\STATE $\mathcal{S}$ = $LSH Sampling$($R$) (Section~\ref{sec:KLParametrizedLSH})
                \STATE Get $\mathcal{T}_{match}$ be the linked pairs (duplicate entities) in $\mathcal{T}$
				\STATE $p$ = $\frac{|\mathcal{T}_{match} \cap \mathcal{S}||}{|\mathcal{T}_{match}|}$
                \STATE Query every pair in $\mathcal{S}$ for match/mismatch (get actual labels). (Graph $G'$)
				\STATE $n_1', n_2', n_3'... n_M'$ = $Traverse(G')$
				\STATE $LSHE$ = $Equation$ \ref{main} ($p$ , $n_1', n_2', n_3'... n_M'$)
			\end{algorithmic}
		\end{algorithm}
	\end{minipage}
	\caption{Overview of our proposed unique entity estimation algorithm.}
\label{samplep}
	\end{figure}
%\end{wrapfigure}

\section{Experiments}
\label{experiments}

We evaluate the effectiveness of our proposed methodology on the Syrian data set and three additional real data sets, where the Syrian data set is only partially labeled, while the other three data sets are fully labeled. We first perform evaluations and comparisons on the three fully labeled data sets, and then give an estimate of the documented number of identifiable victims for the Syrian data set.

\looseness=-1
\begin{itemize}
    \item \textbf{Restaurant}: The \textbf{Restaurant} data set contains 864 restaurant records collected from Fodor's and Zagat's restaurant guides.\footnote{Originally provided by Sheila Tejada, downloaded from http://www.cs.utexas.edu/users/ml/riddle/data.html.}  There are a total of 112 duplicate records. Attribute information contains name, address, city, and cuisine.

        \item \textbf{CD}: The \textbf{CD} data set that includes 9,763 CDs randomly extracted from freeDB.\footnote{https://hpi.de/naumann/projects/repeatability/datasets/cd-datasets.html.} There are a total of 299 duplicate records. Attribute information consists of 106 total features such as artist name, title, genre, among others.

    \item \textbf{Voter}: The  \textbf{Voter} data has been scraped and collected by \cite{christen14voter} beginning in October 2011. We work with a subset of this data set containing 324,074 records. There are a total of 68,627 duplicate records. Attribute information contains personal information on voters from North Carolina including full name, age, gender, race, ethnicity, address, zip code, birth place, and phone number.

    \item \textbf{Syria}: The \textbf{Syria} data set comprises data from the Syrian conflict, which covers the same time period, namely, March 2011 -- April 2014. This data set is not publicly available and was provided by HRDAG. The respective data sets come from the Violation Documentation Centre (VDC), Syrian Center for Statistics and Research (CSR-SY), Syrian Network for Human Rights (SNHR), and Syria Shuhada website (SS). Each database lists a different number of recorded victims killed in the Syrian conflict, along with available identifying information including full Arabic name, date of death, death location, and gender.\footnote{These databases include documented identifiable victims and not those who are missing in the conflict. Hence, any estimate reported only refers to the data at hand.}
\end{itemize}

The above datasets cover a wide spectrum of different varieties observed in practice. For each data set, we report summary information in
Table~\ref{table1}.

\begin{table}[t]
	\centering	
	\tiny
%	\scriptsize	
	\begin{tabular}{lccccc}
		\toprule
		\textbf{DBname} & \textbf{Domain}& \textbf{Size} & \textbf{\# Matching Pairs} & \textbf{\# Attributes}& \textbf{\# Entities} \\
		\midrule
		Restaurants & Restaurant Guide&864  &112 & 4&752 \\
	
		CD & Music CDs&9,763  &299 & 106 & 9,508\\
	
		Voter & Registration Info&324,074  & 70,359 & 6 & 255,447\\
	
		Syria & Death Records&354,996  & N/A & 6 & N/A\\
		\bottomrule
	\end{tabular}	
	\caption{ We present five important features of the four data sets.
%	(three well-labeled real data set and Syrian data set).
	\textbf{Domain} reflects the variety of the data type we used in the experiments. \textbf{Size} is the number of total records respectively. \textbf{\# Matching Pairs} shows how many pair of records point to the same entity in each data set. \textbf{\# Attributes} represents the dimensionality of individual record. \textbf{\# Entities} is the number of unique records. }
	\label{table1}
\end{table}

\begin{figure}[H]
	%\begin{figure}[!Htb]
	%	\centering
	\mbox{
		
		\includegraphics[width=\linewidth]{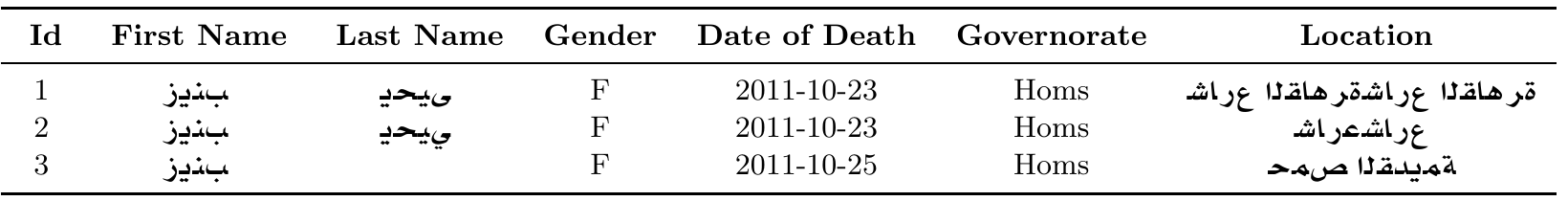}
	}
	\caption{We show several death records in Syrian dataset from VDC, which allows for public access to some of the data. All of the three records belong to the same entity, labeled by human experts. Record 1 and 2 are similar in all attributes while Record 1 and 3 are very different. Due to the variation in the data, records that are very similar are likely to be linked as the same entity, however, it is more difficult to make decisions when records show differences, such as record 1 and 3. This illustrates some of the limitations from deterministic blocking methods discussed in Section \ref{sec:WHYLSH}.}
	\label{fg1}
\end{figure}

\subsection{Evaluation Settings}
In this section, we outline our evaluation settings. We denote Algorithm~\ref{Est} as the LSH Estimator (LSHE). We make comparisons to the non-adaptive variant of our estimator (PRSE), where we use plain random sampling (instead of adaptive sampling). This baseline uses the same procedure as our proposed LSHE, except that the sampling is done uniformly. A comparison with PRSE quantifies the advantages of the proposed adaptive sampling over random sampling. In addition, we implemented the two other known sampling methods, for connected component estimation, proposed in \cite{1978paper} and \cite{chazelle2005approximating}. For convenience, we denote them as Random Sub-Graph based Estimator (RSGE), and BFS on Random Vertex based Estimator (BFSE) respectively. Since the algorithms are based on sampling (adaptive or random), to ensure fairness, we fix a budget $m$ as the number of pairs of vertices considered by the algorithm. Note that any query for an edge is a part of the budget. If the fixed budget is exhausted, then we stop the sampling process and use the corresponding estimate, using all the information available.

We briefly describe the implementation details of the four considered estimators below:
\begin{enumerate}
\item \textbf{LSHE:} In our proposed algorithm, we use the ($K, L$) parameterized LSH algorithm to generate samples of record pairs using Algorithm~\ref{samplep}, where recall $K$ and $L$ control the resulting sample size (section~\ref{sec:parameterstudy}). Given $K, L$ as an input to Algorithm \ref{Est}, we use the sample size as the value of the fixed budget $m$. Table \ref{table2} gives different sample budget sizes (with the corresponding $K$ and $L$) and corresponding values of $p$ for selected samples in three real data sets.
%Parameters $K$ and $L$ controls the size of the sample size. Please refer to ~\ref{sec:parameterstudy} for details. We also report this values of $K$ and $L$ for all the datasets.  Table \ref{table2} gives different sample budget sizes and corresponding values of $p$ for selected samples in three real data sets.
\item \textbf{PRSE:} For a fair comparison, in this algorithm, we randomly sample the same number of record pairs used by LSHE. We then perform the same estimation process as LSHE but instead use $p = \dfrac{2m}{M(M-1)},$ which corresponds to the random sampling probability to get the same number of samples, which is $m$.
\item \textbf{RSGE~\citep{1978paper}:} This algorithm requires performing breadth first search (BFS) on each randomly selected vertices. BFS requires knowing all edges (neighbors) of a node for the next step, which requires $M-1$ edge queries. To ensure the fixed budget $m$, we end the traversal when the number of distinct edge queries reaches the fixed budget $m$.

\item \textbf{BFSE~\citep{chazelle2005approximating}:} This algorithm samples a subgraph and observes it completely. This requires labeling all the pairs of records in the sampled sub-graph. To ensure same budget $m$, the sampled sub-graph has approximately $\sqrt{2m}$ vertices.
\end{enumerate}

\textbf{Remark}: To the best of our knowledge there have been no experimental evaluations of the two algorithms of \cite{1978paper} and \cite{chazelle2005approximating} in the literature. Hence, our results could be of independent interest in themselves.

We compute the relative error (RE), calculated as $$\text{RE} = \dfrac{|\text{LSHE} - n|}{n},$$ for each of the estimators, for different values of the budget $m$. We plot the $\text{RE}$ for each of the estimators, over a range of values of $m$, summarizing the results in Figure~\ref{fg1}.

All the estimators require querying pairs of records compared to labeled ground truth data for whether they are a match or a non-match. As already mentioned, in the absence of full labeled ground truth data, we can use a supervised classifiers such as SVMs as a proxy, assuming at least some small amount of labeled data exists. By training an SVM, we can use this as a proxy for labeled data as well. We use such a proxy in the Syrian data set because we are not able to query every pair of records to determine whether they are true duplicates or not.

%manual expert labelers for labeling each query, we use supervised classifiers such as support vector machine (SVM) as a proxy. Note that we have access to a small sample of manually labeled records. We can train an SVM, or other supervised classifiers, which we can be used as a proxy for manual labelers. This proxy is required for the Syrian dataset where we cannot query every pair for match-mismatch.

We start with the three data sets where fully labelled ground truth data exists. For LSHE, we compute the estimation accuracy using both the supervised SVM (Section~\ref{syria}) as well as using the fully labelled ground truth data. The difference in these two numbers quantifies the loss in estimation accuracy due to the use of the proxy SVM prediction instead of using ground truth labeled data. In our use of SVMs, we take less than 0.01$\%$ of the total number of the possible record pairs as the training set.

\begin{figure}[H]
%\begin{figure}[!Htb]
	%	\centering
	\mbox{
		\hspace{-0.1in}
		\includegraphics[width=.36\linewidth]{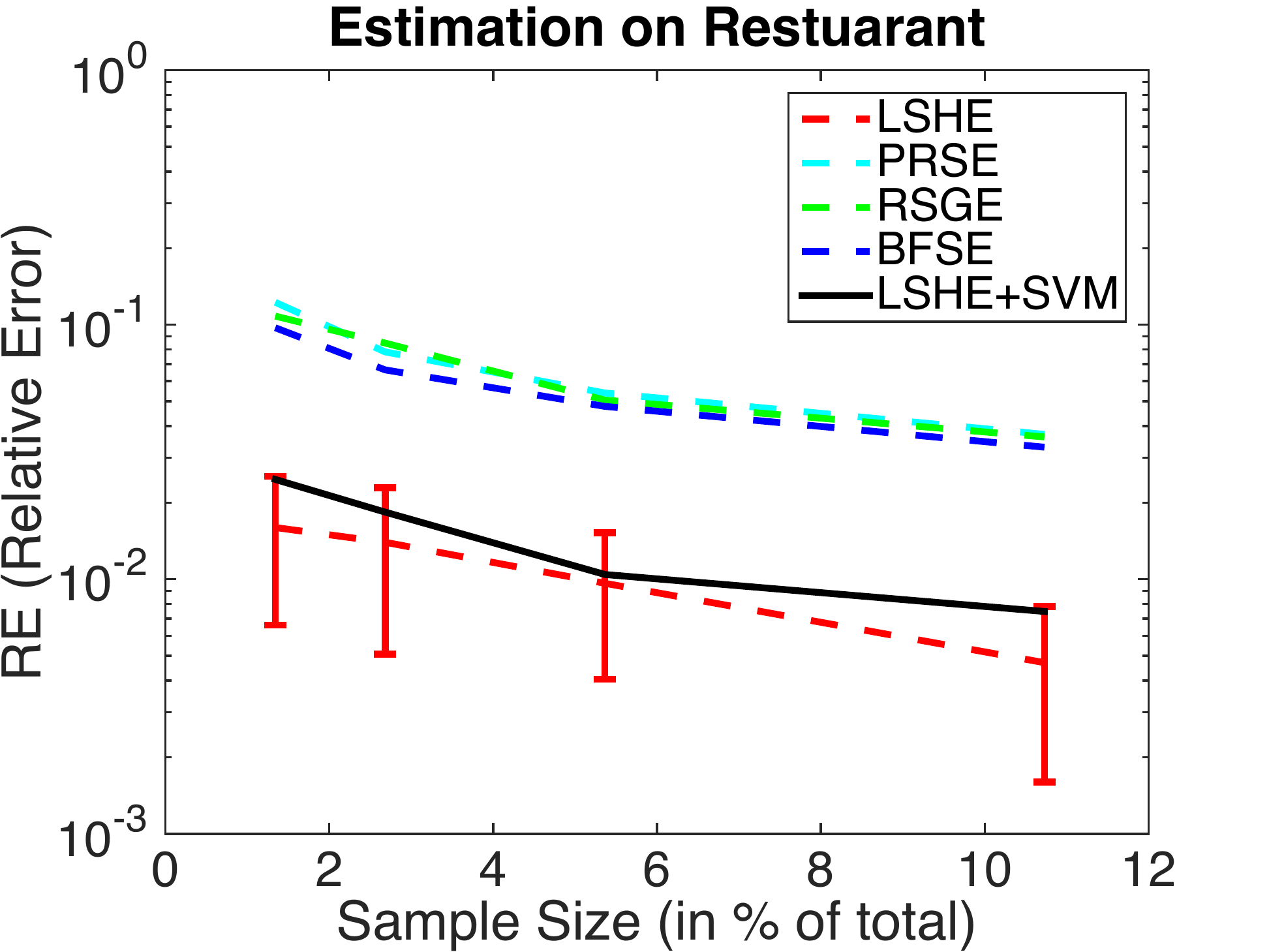}\hspace{-0.18in}
		\includegraphics[width=.36\linewidth]{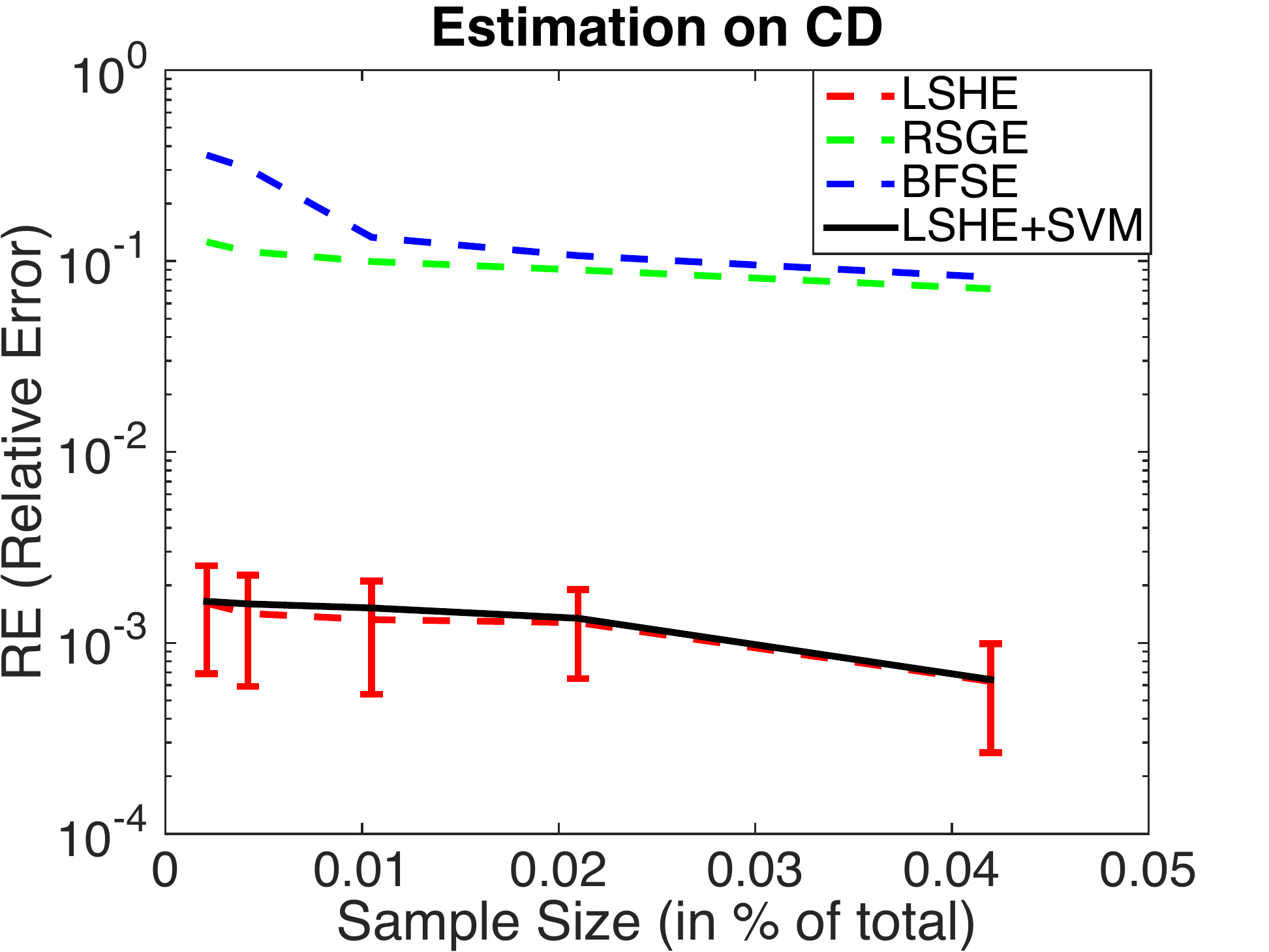}\hspace{-0.16in}
		\includegraphics[width=.36\linewidth]{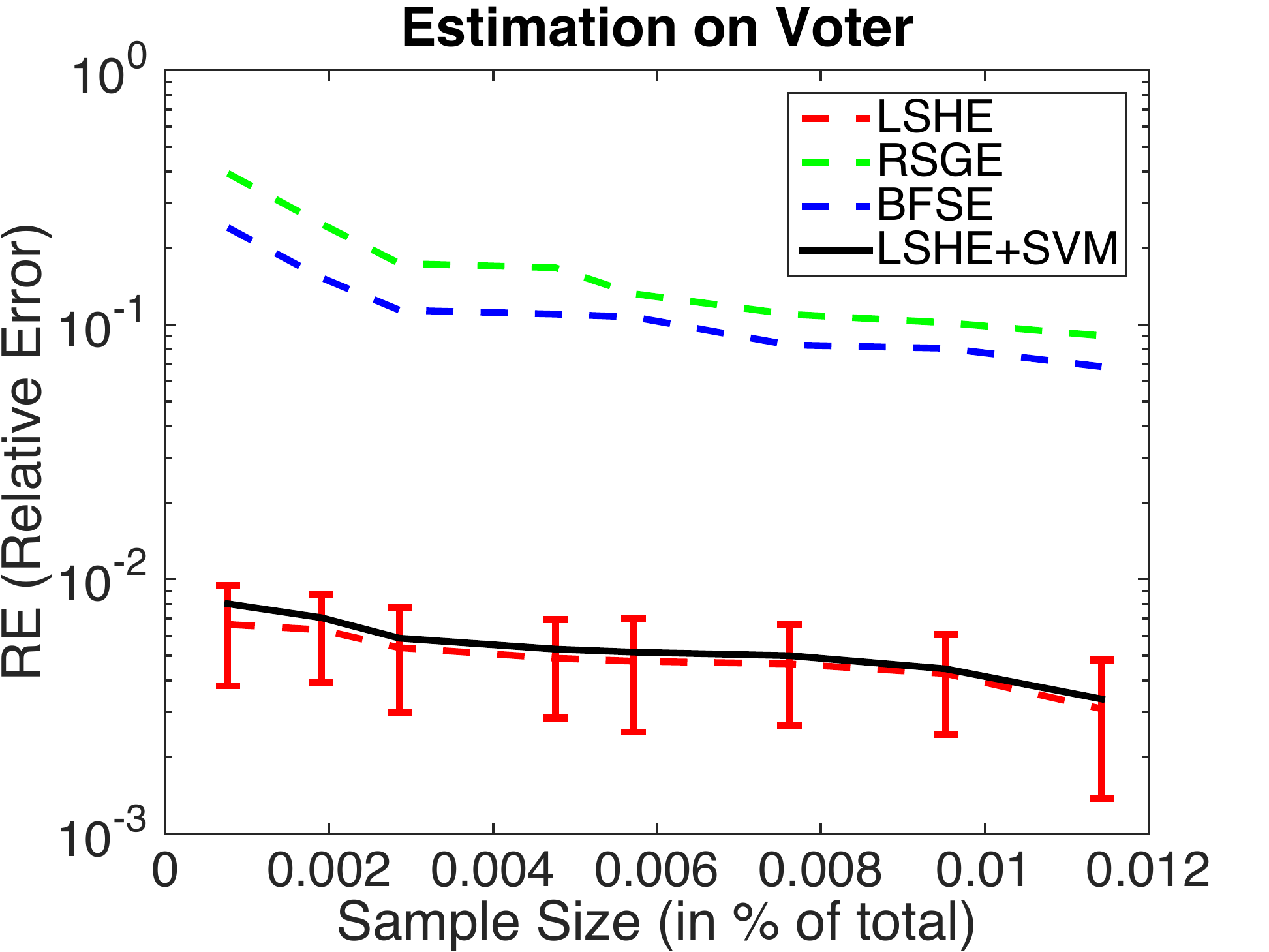}}
	\caption{ The dashed lines show the $\text{RE}$ of the four estimators on the three real data sets, where the y-axis is on the log-scale. Observe that LSHE outperforms all other three estimators in one to two orders of magnitude. The standard deviation of the $\text{RE}$ for LSHE is also shown in the plots with the red error bars, which is with respect to randomization of hash functions.  In particular, the PRSE performs unreliable estimation on the CD and Voter data sets. The dashed and solid black lines represent RE of LSHE using ground truth labels and a SVM classifier (y-axis is on the log scale). We discuss the LSHE + SVM estimator in section \ref{syria} (solid black line).}
	\label{fg1}
\end{figure}

\subsection{Evaluation Results}
In this section, we summarize our results regarding the aforementioned evaluation metrics by varying the sample size $m$ on the three real data sets (see Figure \ref{fg1}). We notice that for the CD and Voter data sets, we cannot obtain any reliable estimate (for any sample size) using PRSE. Recall that plain random sampling almost always samples pairs of records that correspond to non-matches. Thus, it is not surprising that this method is unreliable because sampling random pairs is unlikely to result in a duplicate pair for entity resolution tasks. Even with repeated trials, there are no edges in the specified sampled pairs of records, leading to an undefined value of $p$. This phenomenon is a common problem in random sampling estimators over sparse graphs. Almost all the sampled nodes are singletons. Subsampling a small sub-graph leads to a graph with most singleton nodes, which leads to a poor accuracy of BFSE. Thus, it is expected that random sampling will perform poorly. Unfortunately, there is no other baseline for unbiased estimation of the number of unique entities.

From Figure \ref{fg1} observe that the $\text{RE}$ for proposed estimator LSHE is approximately one to two orders of magnitude lower than the other considered methods, where the  y-axis is on the log-scale. Undoubtedly, our proposed estimator LSHE consistently leads to significantly lower $\text{RE}$ (lower error rates) than the other three estimators. This is not surprising from the analysis shown in section~\ref{sec:randVsAdaptive}. The variance of random sampling based methodologies will be significantly higher.

%to extreme sparsity in the data (random pairs being duplicate has almost zero probability).

Taking a closer look at LSHE, we notice that we are able to efficiently generate samples with very high values of $p$ (see Table~\ref{table2}).
%{\color{red} We need a line describing the values of $p$ and referring to our precision recall experiments.}
%{\color{green} Here I suggest not relating to the precision recall experiments because they are in next section and we haven't explain those terms.}
In addition, we can clearly see that LSHE achieves high accuracy with very few samples. For example, for the CD data set, with a sample size less than $0.05\%$ of the total possible pairs of records of the entire data set, LSHE achieves $0.0006$ $\text{RE}$. Similarly, for the Voter data set, with a sample size less than $0.012\%$ of the total possible pairs of records of the entire data set, LSHE achieves $0.003$ $\text{RE}$.

Also, note the small values of $K$ and $L$ parameters required to achieve the corresponding sample size. $K$ and $L$ affect the running time, and small values $KL \ll O(M^2)$ indicate significant computational savings as argued in section~\ref{sec:CC}

As mentioned earlier, we also evaluate the effect of using SVM prediction as a proxy for actual labels with our LSHE. The dotted plot shows those results. We remark on the results for LSHE + SVM in section~\ref{syria}.

\begin{table}[ht]
	\centering
	\setlength\tabcolsep{0.18cm}
		\begin{tabular}{l|cccc|cccc|cccc}
			\toprule
			\multicolumn{4}{c}{\textbf{Restaurant}}&\multicolumn{4}{c}{\textbf{CD}}&\multicolumn{4}{c}{\textbf{Voter}}
			\\
			\midrule
			Size&1.0&2.5&5.0&10&0.005 & 0.01& 0.02 & 0.04&0.002&0.006&0.009&0.013
			\\
			
			$p$&0.42&0.54&0.65&0.82& 0.72& 0.74& 0.82& 0.92& 0.62&0.72& 0.76& 0.82
			\\
			
			$K$&1&1&1&1& 1& 1& 1& 1& 4& 4& 4& 4
			\\						
			
			$L$& 4 & 8 &12&20& 5& 6 & 8 & 14 &  25&  32&  35&40
			\\

			\bottomrule			
		\end{tabular}
			\caption{ We illustrate part of the sample sizes (in \% in TOTAL) for different sets of samples generated by Min-Wise Hashing and their corresponding $p$ in all three data sets. }
			\label{table2}
\end{table}

\section{Documented Identifiable Deaths in the Syrian Conflict}
\label{syria}
In this section, we describe how we estimate the number of documented identifiable deaths for the Syrian data set. As noted before, we do not have ground truth labels for all record pairs, but the data set was partially labeled with 40,000 record pairs (out of 63 billion). We propose an alternative (automatic) method of labeling the sample pairs, which is also needed by our proposed estimation algorithm. More specifically, using the partially labeled pairs, we train an SVM. In fact, other supervised methods could be considered here, such as random forests, Bayesian Adaptive Regression Trees (BART), among others, however, given that SVMs perform very well, we omit such comparisons as we expect the results to be similar if not worse.

To train the SVM, we take every record pair and generate $k$-grams representation for each record. Then we spilt the partially labeled data into training and testing sets, respectively. Each training and testing set contains a pair of records $x_k = [R_i,\ R_j]$. In addition, we can use a binary label indicating whether the record pair is a match or non-match. That is, we can write the data as $\{x_k =[R_i,R_j],y_k\}$ as the set difference of the $k$-grams of the strings of pairs of records $R_i$ and $R_j$, respectively. Observe that $y_k = 1$ if the $R_i$ and $R_j$ is labelled as match and $y_k = -1$ otherwise. Next, we tune the SVM hyper-parameters using 5-fold cross-validation, and we find the accuracy of SVM on the testing set was 99.9\%.  With a precision as high a 0.99, we can reliably query an SVM and now treat this as an expert label.

% We can use $R_1$, $R_2$ as well} and a binary classifier indicating whether the record pair is a match or non-match. \textcolor{red}{AS: THIS NEXT PART IS NOT CLEAR AND NEEDS TO BE WRITTEN MORE CLEARLY WITH WORDS AND SYMBOLS. I RECOMMEND WE REMOVE THE NEXT SENTENCE. I THINK IT WILL CONFUSED THE READER.}
%
%
% That is, the training data $\{x_i=[R_a,R_b],y_i\}$ has $x_i$ as the set difference of the $n$-gram features of the pairs of records $R_a$ and $R_b$ respectively.  Observe that $y_i = 1$ if the $R_a$ and $R_b$ is labelled as match and $y_i = -1$ otherwise. Next, we tune the SVM hyper-parameters using 5-fold cross validation, and we find the the SVM of the testing set was 99.9\%. \textcolor{red}{BC: WHAT was the split regarding testing and training? This is something that someone might ask regarding overfitting and actually I want to know more about. We should have a skype call about this and perhaps run more experiments after the submission.} With a precision as high a 0.99, we can reliably query an SVM and now treat this an expert label.

%We tune the SVM hyper parameter using standard 5-fold cross validations. We found that the accuracy of SVM on the unseen test set was over 99.9\%. With this high precision, we can reliably query an SVM and treat that as an expert label.

To understand the effect of using SVM prediction as a proxy to label queries in our proposed unique entity estimation algorithm, we return to observing the behavior in figure \ref{fg1}. We treat the LSHE estimator on the other three real datasets as our baselineand compare to LHSE with the SVM component, where the SVM prediction replaces the querying process (LSHE +SVM). Observe in  Figure~\ref{fg1}, that the plot for LSH (solid black line) and LSH+SVM (dotted black line) overlap indicating a negligible loss in performance. This overlap is expected given the high accuracy (high precision) of the SVM classifier.

%Here, we compare to baseline comparison of the LSHE estimator on the other three real datasets which have fully labeled ground truth data. Next, we compare to a new variant of LHSE with an SVM component, where the SVM prediction replaces the querying process. We refer to this variant as LSHE +SVM. Observe in  Figure~\ref{fg1}, that the plot for LSH (solid black line) and LSH+SVM (dotted black line) overlap indicating a negligible loss in performance. This overlap is expected given the high accuracy (high precision) of the SVM classifier.

%We use the LSHE estimator on the other three real datasets which have fully labelled ground truth data.  We then compare the original procedure, which relies on ground truth labeling of samples, with LSHE where the SVM prediction replaces the querying part. We call the method with the SVM proxy as LSHE +SVM. Observe in  Figure~\ref{fg1}, that the plot for LSH (solid black line) and LSH+SVM (dotted black line) overlap indicating a negligible loss in performance. This overlap is expected given the high accuracy (high precision) of SVM classifier.

\subsection{Running Time}
We briefly highlight the speed of the sampling process since it could be used for on the fly or online unique entity estimation.
The total running time for producing 450,000 sampled pairs (out of a possible 63 billion) used for the LSH sampler (Section~\ref{sec:KLParametrizedLSH}) with $K=15$ and $L=10$ is 127 seconds. On the other hand, it will take approximately take 8 days to compute all pairwise similarities across the 354,996 Syrian records. Computing all pairwise similarities is the first prerequisite for any known adaptive sampling over pairs based on similarity -- if we are not using the proposed LSH sampler.
(Note: there are other ways of blocking~\citep{christen_2012}, however as mentioned in Section~\ref{sec:WHYLSH} they are mostly deterministic (or rule-based) and do not provide an estimator of the unique entities.)

\subsection{Unique Number of Documented Identifiable Victims}

In the Syrian dataset, with 354,996 records and possibly 63 billion ($6.3 \times 10^{10}$) pairs, our motivating goal was to estimate the unique number of documented identifiable victims. Specifically, in our final estimate, we use
452,728 sampled pairs that are given by  LSHE+SVM  ($K = 15$, $L = 10$) which has approximately $p=0.83$ based on the subset of labeled pairs. The sample size was chosen to balance the computational runtime and the value of $p$.
Specifically, one wants high values of $p$ (for a resulting low variance of our estimate) and, to balance running time, we
limit the sample size to be around the total number of records $O(M)$, to ensure a near linear time algorithm. (Such settings are determined by the application, but as we have demonstrated they work for a variety of real entity resolution data sets). We chose the SVM as our classifier to label the matches and non-matches. The final unique number of documented identifiable victims in the Syrian data set was estimated to be 191,874$ \pm 1772$, very close to the 191,369 documented identifiable deaths reported by HRDAG 2014, where their process is described in Appendix~\ref{sec:syrian}.

\subsection{Effects of $L$, $K$, on sample size and $p$}
\label{sec:parameterstudy}
%\textcolor{red}{RS: Smooth this out and think about moving this to the appendix. Maybe call this a sensitivity analysis.}
In this section, we discuss the sensitivity of our proposed method as we vary the choice of $L$, $K$, the sample size $M$, and $p$.

We want both $KL \ll M$ as well as the number of samples to be $\ll M^2$, for the process to be truly sub-quadratic. For accuracy, we want high values of $p$, because the variance is monotonic in $p,$ which is also the recall of true labeled pairs. Thus, there is a natural trade-off. If we sample more, we get high $p$ but more computations.

$K$ and $L$ are the basic parameters of our sampler (Section~\ref{sec:KLParametrizedLSH}), which provide a tradeoff between the computationally complexity and accuracy. A large value of $K$ makes the buckets sparse exponentially), and thus, fewer pairs of records are sampled from each table. A large value of $L$ increases the repetition of hash tables (linearly), which increases the sample size. As already argued, the computational cost is $O(MKL)$.

To understand the behavior of $K$, $L$, $p$, and the computational cost, we perform a set of experiments on the Syrian dataset.  We use n-gram of 2---5, we vary L from 5--100 by steps of 5 and K takes values 15,18,20,23,25,28,30,32,35. For all these combinations, we then plot the recall (also the value of $p$) and the reduction ratio (RR), which is the percentage of computational savings. \textcolor{black}{A 99\% reduction ratio means that the original space has been reduced to only having to look at a
only 1\% of total sampled pairs.} Figure~\ref{syria-takethatAssad} shows the tradeoffs between reduction ratio and recall (or value of $p$). Every dot in the figures is one whole experiment.

Regardless of the n-gram variation from 2--5, the recall and reduction ratio (RR) are close to 1 as illustrated in Figure~\ref{syria-takethatAssad}. We see that an n-gram of 3 overall is most stable in having a recall and RR close to 0.99. We observe that $K=15$ and $L=10$ gives a high recall of around 83\% with less than half a million pairs (out of 63 billion possible) to evaluate ($RR \ge 0.99999$).

\begin{figure}[htbp]
\begin{center}
\includegraphics[width=\textwidth]{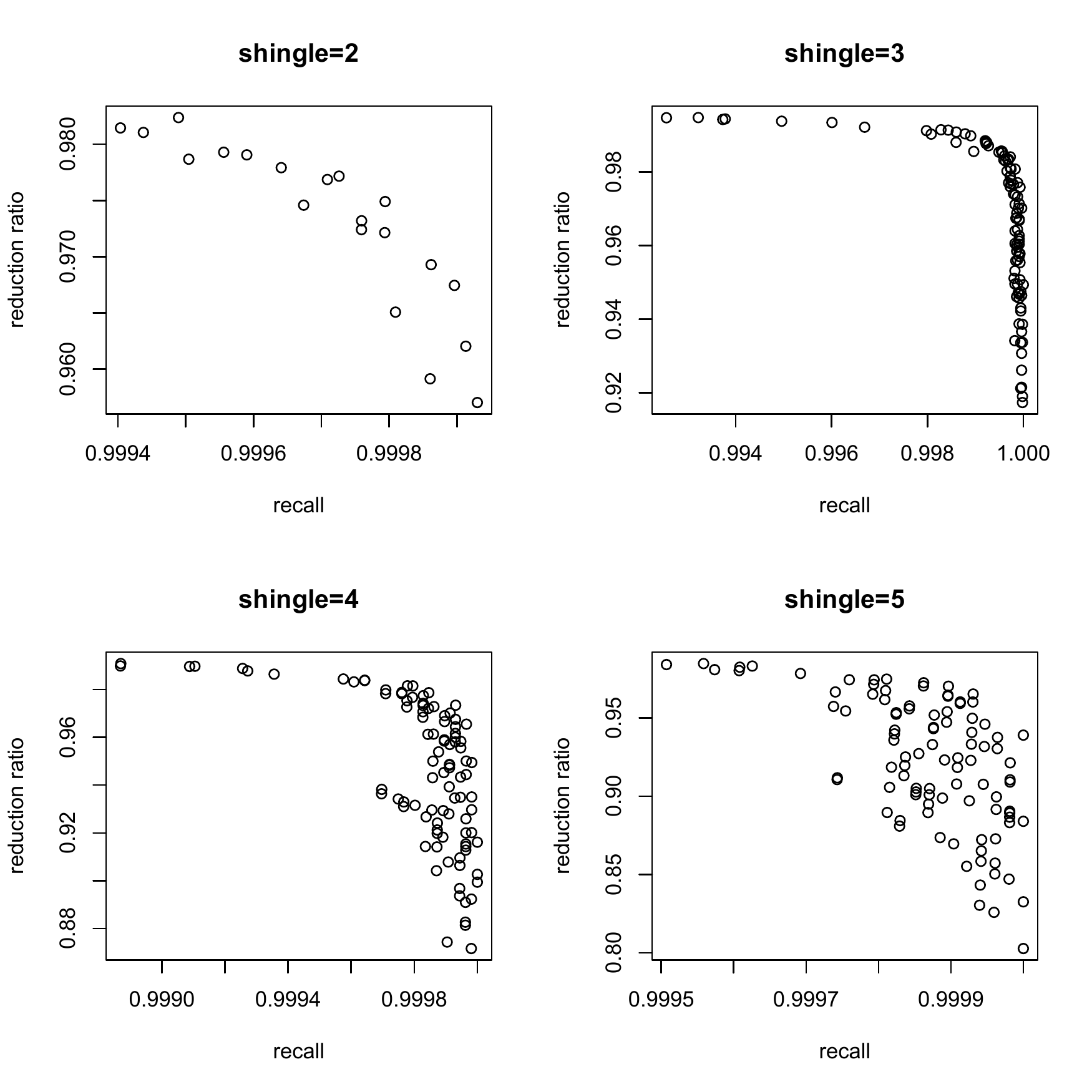}
\caption{For shingles 2--5, we plot the RR versus the recall. Overall, we see the best behavior for a shingle of 3, where the RR and recall can be reached at 0.98 and 1, respectively. We allow L and K to vary on a grid here. L varies from 5--100 by steps of 5; and K takes values 15, 18, 20, 23, 25, 28, 30, 32, and 35.}
\label{syria-takethatAssad}
\end{center}
\end{figure}

\section{Discussion}
Motivated by three real entity resolution tasks and the ongoing Syrian conflict, we have proposed a general, scalable algorithm for unique entity estimation. Our proposed method is an adaptive LSH on the edges of a graph, which in turn estimates the connected components in sub-quadratic time. Our estimator is unbiased and has provably low variance in contrast to other such estimators for unique entity estimation in the literature. In experimental results, it outperforms other estimators in the literature on three real entity resolution data sets. Moreover, we have estimated the number of documented identifiable deaths to be 191,874$ \pm 1772$, which very closely matches the 2014 HRDAG estimate, completed by hand-matching. To our knowledge, we have the first estimate for the number of documented identifiable deaths with a standard error associated with such an estimate. Our methods are scalable, potentially bringing impact to the human rights community, where such estimates could be updated in near real time. It could lead to further impact in public policy and transitional justice in Syria and other areas of conflict globally.

%where Syria has currently been at the forefront of both national and international news.
%Add this later.
%
%We should talk about doing some kind of interaactive learning approach with this method that would allow for the training data to be improved. This would be super cool and be great for future work! (Or maybe let's not mention it and work on it).

\textbf{Acknowledgements}: We would like to thank the Human Rights Data Analysis Group (HRDAG) and specifically, Megan Price, Patrick Ball, and Carmel Lee for commenting on our work and giving helpful suggestions that have improved the methodology and writing. We would also like to thank Stephen E. Fienberg and Lars Vilhuber for making this collaboration possible. PhD student Chen is supported by National Science Foundation (NSF) grant number 1652131. Shrivastava's work is supported by NSF-1652131 and NSF-1718478. Steorts's work is supported by NSF-1652431, NSF-1534412, and the Laboratory for Analytic Sciences (LAS). This work is representative of the author's alone and not of the funding organizations.
\clearpage
\newpage

%\begin{supplement}[id=suppA]
%  \sname{Appendix}
%  \stitle{Supplementary Material for ``Regularized Brain Reading with Shrinkage and Smoothing''}
%  \slink[doi]{COMPLETED BY THE TYPESETTER}
%  \sdatatype{.pdf}
%  \sdescription{}
%\end{supplement}
%
%
%\input{appendix_journal_v2}

\clearpage
\newpage

\appendix

\section{Syrian Data Set}
\label{sec:syrian}
In this section, we provide a more detailed description about the Syrian data set. As mentioned in section \ref{syriadata}, via
collaboration with the Human Rights Data Analysis Group (HRDAG), we have access to four databases. They come from the Violation Documentation Centre (VDC), Syrian Center for Statistics and Research (CSR-SY), Syrian Network for Human Rights (SNHR), and Syria Shuhada website (SS). Each database lists each victim killed in the Syrian conflict, along with identifying information about each person (see \cite{price_2013} for further details).

Data collection by these organizations is carried out in a variety of ways. Three of the groups (VDC, CSR-SY, and SNHR) have trusted networks on the ground in Syria.  These networks collect as much information as possible about the victims. For example, information is collected  through direct community contacts. Sometimes information comes from a victim's friends or family members. Other times, information comes from religious leaders, hospitals, or morgue records.  These networks also verify information collected via social and traditional media sources.  The fourth source, SS, aggregates records from multiple other sources, including NGOs and social and traditional media sources (see \url{http://syrianshuhada.com/} for information about specific sources).

These lists, despite being products of extremely careful, systematic data collection, are not probabilistic samples \citep{Sig2015, IOAS2015, CJLS2015, price2014updated}. Thus, these lists cannot be assumed to represent the underlying population of all victims of conflict violence.  Records collected by each source are subject to biases, stemming from a number of potential causes, including a group's relationship within a community, resource availability, and the current security situation.  Although it is beyond the scope of this paper, final analyses of these sources must appropriately adjust for such biases before drawing conclusions about patterns of violence.

\subsection{Syrian Handmatched Data Set}
\label{app:hand}
We describe how HRDAG's training data on the Syrian data set was created, which we use in our paper. We would like to note that we only use a small fraction of the training data for two reasons. The first is so that we can see how close our estimate is in practice to their original handmatched estimate, given that such a large portion of the data was handmatched. Second, we want to avoid using too much training data to avoid biases and also because such handmatching efforts would not be possible moving forward as the Syrian conflict continues, and our small training data set is meant for one moving forward in practice.
%created the handmatched data set used in our paper.

First, all documented deaths recorded by any of the documentation groups were concatenated together into a single list. From this list, records were broadly grouped according to governorate and year. In other words, all killings recorded in Homs in 2011 were examined as a group, looking for records with similar names and dates.

Next, several experts review these ``blocks", sometimes organized as pairs for comparison and other times organized as entire spreadsheets for review.  These experts determine whether pairs or groups of records refer to the same individual victim or not. Pairs or groups of records determined to refer to the same individual are assigned to the same ``match group." All of the records contributing to a single ``match group" are then combined into a single record. This new single record is then again examined as a pair or group with other records, in an iterative process.

For example, two records with the same name, date, and location may be identified as referring to the same individual, and combined into a single record. In a second review process, it may be found that that record also matches the name and location, but not date, of a third record. The third record may list a date one week later than the two initial records, but still be determined to refer to the same individual. In this second pass, information from this third record will also be included in the single combined record.

When records are combined, the most precise information available from each of the individual records is kept. If some records contain contradictory information (for example, if records A and B record the victim as age 19 and record C records age 20) the most frequently reported information is used (in this case, age 19). If the same number of records report each piece of contradictory information, a value from the contradictory set is randomly selected.

Three of the experts are native Arabic speakers; they review records with the original Arabic content. Two of the experts review records translated into English. These five experts review overlapping sets of records, meaning that some records are evaluated by two, three, four, or all five of the experts. This makes it possible to check the consistency of the reviewers, to ensure that they are each reaching comparable decisions regarding whether two (or more) records refer to the same individual or not.

After an initial round of clustering, subsets of these combined records were then re-examined to identify previously missed groups of records that refer to the same individual, particularly across years (e.g., records with dates of death 2011/12/31 and 2012/01/01 might refer to the same individual) and governorates (e.g., records with neighboring locations of death might refer to the same individual).

\section{Unique Entity Estimation Proofs}
\label{sec:app}
First, we introduce four indicators.
First, let $\mathbb{I}_2$ denote every 2-vertex clique in $G^*$ (recall that $G^*$ is the original graph and $G'$ is the observed one):
	\begin{equation}
	\mathbb{I}_2=\begin{cases}
	1, & \text{if this clique is in } G'.\\
	0, & \text{otherwise}.
	\end{cases}
	\end{equation}
Second, let 	
$\mathbb{I}_{33}$ denote every 3-vertex clique in $G^*$:
	\begin{equation}
	\mathbb{I}_{33}=\begin{cases}
	1, & \text{if this clique remains as a 3-clique in } G'.\\
	0, & \text{otherwise}.
	\end{cases}
	\end{equation}
Third, let
$\mathbb{I}_{32}$ denote every 3-vertex clique in $G^*$:
	\begin{equation}
	\mathbb{I}_{32}=\begin{cases}
	1, & \text{if this clique breaks to a 2-clique in } G'.\\
	0, & \text{otherwise}.
	\end{cases}
	\end{equation}
Finally, let
$\mathbb{I}_{31}$ denote every 3-vertex clique in $G^*$:
	\begin{equation}
	\mathbb{I}_{31}=\begin{cases}
	1, & \text{if this clique breaks into only 1-cliques in } G'.\\
	0, & \text{otherwise}.
	\end{cases}
	\end{equation}

\subsection{Expectation}
We now prove that our estimator is unbiased. Consider
\begin{equation}
\label{C3}
\mathbb{E}[n_3']= \mathbb{E}[\sum_{i=1}^{n_3^*}{\mathbb{I}_{33i}}]= n_3^* \cdot p^2 \cdot (3-2p),
\end{equation}
\begin{equation}
\label{C2}
\begin{split}
\mathbb{E}[n_2'] &=    \mathbb{E}[\sum_{i=1}^{n_2^*}{\mathbb{I}_{2i}}] + \mathbb{E}[\sum_{i=1}^{n_3^*}{\mathbb{I}_{32i}}]     \\
&= n_2^* \cdot p + n_3^*\cdot (3\cdot(1-p)^2\cdot p), \quad \text{and}
\end{split}
\end{equation}
\begin{equation}
\label{C1}
\begin{split}
\mathbb{E}[n_1']&= n_1^*+
\mathbb{E}[\sum_{i=1}^{n_2^*}{(1-\mathbb{I}_{2i})}] + \mathbb{E}[\sum_{i=1}^{n_3^*}{\mathbb{I}_{32i}}] + \mathbb{E}[\sum_{i=1}^{n_3^*}{\mathbb{I}_{31i}}]\\
&= n_1^* + n_2^* \cdot(2\cdot(1-p)) + n_3^* \cdot (3\cdot (1-p)^2\cdot p) \\
&+ n_3^* \cdot (3\cdot(1-p)^3 ).
\end{split}
\end{equation}

Our estimator is unbiased via equations \ref{C1}, \ref{C2}, \ref{C3}:
\begin{equation*}
\begin{split}
\mathbb{E}[LSHE] =&\mathbb{E}[n_1' +  n_2' \cdot \frac{2p-1}{p}  \\
&+ n_3' \cdot \frac{1-6 \cdot (1-p)^2 \cdot p}{p^2 \cdot (3-2p) }
+ \sum_{i=4}^{M} n_i] \\
& = \mathbb{E}[n_1'] +  \frac{2p-1}{p} \cdot \mathbb{E}[n_2']    \\
&+ \frac{1-6 \cdot (1-p)^2 \cdot p}{p^2 \cdot (3-2p)} \cdot \mathbb{E}[n_3']   +  \mathbb{E}[\sum_{i=4}^{M}n_i]\\
& = n_1^* + n_2^* + n_3^* + \sum_{i=4}^{N}n_i^*\\
& = n.
\end{split}
\end{equation*}

%\begin{equation}
%\label{C3}
%\mathbb{E}[\hat{n_3}]= \mathbb{E}[\sum_{i=1}^{n_3^*}{\mathbb{I}_{33i}}]= n_3^* \cdot p^2 \cdot (3-2p)
%\end{equation}
%\begin{equation}
%\label{C2}
%\begin{split}
%\mathbb{E}[\hat{n_2}] &=    \mathbb{E}[\sum_{i=1}^{n_2^*}{\mathbb{I}_{2i}}] - \mathbb{E}[\sum_{i=1}^{n_3^*}{\mathbb{I}_{32i}}]     \\
%&= n_2^* \cdot p - n_3^*\cdot (3\cdot(1-p)^2\cdot p)
%\end{split}
%\end{equation}
%
%\begin{equation}
%\label{C1}
%\begin{split}
%\mathbb{E}[\hat{n_1}]&= n_1^*+
%\mathbb{E}[\sum_{i=1}^{n_2^*}{(1-\mathbb{I}_{2i})}] + \mathbb{E}[\sum_{i=1}^{n_3^*}{\mathbb{I}_{32i}}] + \mathbb{E}[\sum_{i=1}^{n_3^*}{\mathbb{I}_{31i}}]\\
%&= n_1^* + n_2^* \cdot(2\cdot(1-p)) + n_3^* \cdot (3\cdot (1-p)^2\cdot p) \\
%&+ n_3^* \cdot (3\cdot(1-p)^3 )
%\end{split}
%\end{equation}
%
%\begin{equation}
%\begin{split}
%&\mathbb{E}[Est] =\mathbb{E}[\hat{C}_1 +  \hat{C}_2 \cdot \frac{2p-1}{p}  \\
%&+ \hat{C}_3 \cdot \frac{1-6 \cdot (1-p)^2 \cdot p}{p^2 \cdot (3-2p) }
%+ \sum_{n=4}^{N}\hat{C}_n] \\
%& = \mathbb{E}[\hat{C}_1] +  \frac{2p-1}{p} \cdot \mathbb{E}[\hat{C}_2]    \\
%&+ \frac{1-6 \cdot (1-p)^2 \cdot p}{p^2 \cdot (3-2p)} \cdot \mathbb{E}[\hat{C}_3]   +  \mathbb{E}[\sum_{n=4}^{N}\hat{C}_n]\\
%& = n_1 + n_2 + n_3 + \sum_{n=4}^{N}C_n
%\end{split}
%\end{equation}
%{\color{red} I will see what is a better way of explaining how we drop $\sum_{n=4}^{N}\hat{C}_n$}

\subsection{Variance}
We now turn to deriving the variance of our proposed estimator, showing that
$$\mathbb{V}[LSHE] = n_3^* \cdot \frac{(p-1)^2 \cdot (3p^2-p+1)}{p^2 \cdot (3-2p)} + n_2^* \cdot \frac{(1-p)}{p}.$$
Consider
\begin{equation}
\label{v1}
\begin{split}
&\mathbb{V}[LSHE]=
\mathbb{V}[\frac{1-6 \cdot (1-p)^2 \cdot p}{p^2 \cdot (3-2p) } \cdot \sum_{i=1}^{n_3^*}{\mathbb{I}_{33i}}\\
&+
\frac{2p-1}{p} \cdot
(\sum_{i=1}^{n_2^*}{\mathbb{I}_{2i}}
+
\sum_{i=1}^{n_3^*}{\mathbb{I}_{32i}})\\
&+
\sum_{i=1}^{n_3^*}{\mathbb{I}_{31i}}
+
\sum_{i=1}^{n_3^*}{\mathbb{I}_{32i}}
+
\sum_{i=1}^{n_2^*}{(1-\mathbb{I}_{2i})}]\\
&=
\mathbb{V}ar[\frac{1-6 \cdot (1-p)^2 \cdot p}{p^2 \cdot (3-2p) } \cdot \sum_{i=1}^{n_3^*}{\mathbb{I}_{33i}}
+
\frac{3p-1}{p} \cdot
\sum_{i=1}^{n_3^*}{\mathbb{I}_{32i}}\\
&+
3 \cdot
\sum_{i=1}^{n_2^*}{\mathbb{I}_{31i}}
-
\frac{1}{p} \cdot
\sum_{i=1}^{n_2^*}{\mathbb{I}_{2i}}].
\end{split}
\end{equation}

Next, we replace $1-6 \cdot (1-p)^2 \cdot p$ by $a,$ and by simplifying equation \ref{v1},
we find
\begin{equation}
%\label{v1}
\begin{split}
&=
\frac{a^2}{(p^2 \cdot (3-2p))^2 } \cdot \mathbb{V}[\sum_{i=1}^{n_3^*}{\mathbb{I}_{33i}}]
 +
\frac{(3p-1)^2}{p^2} \cdot
\mathbb{V}[\sum_{i=1}^{n_3^*}{\mathbb{I}_{32i}}]\\
&+
9 \cdot
\mathbb{V}[\sum_{i=1}^{n_2^*}{\mathbb{I}_{31i}}]
-
\frac{1}{p^2} \cdot
\mathbb{V}[\sum_{i=1}^{n_2^*}{\mathbb{I}_{2i}}]
+
Cov(\sum_{i=1}^{n_3^*}{\mathbb{I}_{33i}}, \sum_{i=1}^{n_3^*}{\mathbb{I}_{32i}})\\
&+
Cov(\sum_{i=1}^{n_3^*}{\mathbb{I}_{33i}}, \sum_{i=1}^{n_3^*}{\mathbb{I}_{31i}})
+
Cov(\sum_{i=1}^{n_3^*}{\mathbb{I}_{32i}}, \sum_{i=1}^{n_3^*}{\mathbb{I}_{31i}}).
\end{split}
\end{equation}

Note that the covariance of $\sum_{i=1}^{n_3^*}{\mathbb{I}_{2i}}$ with any indicator is zero due to independence.
Furthermore, since the indicators are Bernoulli distributed random variables, the variance is easily found. Consider
\begin{equation}
\label{v4}
\begin{split}
\mathbb{V}ar[\sum_{i=1}^{n_3^*}{\mathbb{I}_{33i}}]
=  \frac{a^2 \cdot (1-p^2 \cdot (3-2p))}{p^2 \cdot (3-2p) } \cdot n_3^*
\end{split}
\end{equation}

\begin{equation}
\label{v5}
\begin{split}
\mathbb{V}ar[\sum_{i=1}^{n_3^*}{\mathbb{I}_{32i}}]
= \frac{3 \cdot (3p-1)^2 \cdot (1-p)^2 \cdot (1-3p \cdot (1-p)^2)}{p} \cdot n_3^*
\end{split}
\end{equation}

\begin{equation}
\label{v6}
\begin{split}
\mathbb{V}ar[\sum_{i=1}^{n_3^*}{\mathbb{I}_{31i}}]
= 9 \cdot (1-p)^3 \cdot (1-(1-p)^3) \cdot n_3^*
\end{split}
\end{equation}

\begin{equation}
\label{v7}
\begin{split}
\mathbb{V}ar[\sum_{i=1}^{n_2^*}{\mathbb{I}_{2i}}]
= \frac{(1-p)}{p} \cdot n_2^*
\end{split}
\end{equation}

Using equations \ref{v4} -- \ref{v7}, the covariance simplifies to
\begin{equation*}
\begin{split}
&Cov(\sum_{i=1}^{n_3^*}{\mathbb{I}_{33i}}, \sum_{i=1}^{n_3^*}{\mathbb{I}_{32i}})\\
&=  \mathbb{E}[\sum_{i=1}^{n_3^*}{\mathbb{I}_{33i}} \cdot \sum_{i=1}^{n_3^*}{\mathbb{I}_{32i}}) -  	\mathbb{E}[\sum_{i=1}^{n_3^*}{\mathbb{I}_{33i}}] \mathbb{E}[\sum_{i=1}^{n_2^*}{\mathbb{I}_{32i}}]\\
&=  \sum_{i=1}^{n_3^*}{\sum_{j=1}^{n_3^*}{
		\mathbb{E}[\mathbb{I}_{33i} \cdot \mathbb{I}_{32j}] -  	
		\mathbb{E}[\mathbb{I}_{33i}]
		\mathbb{E}[\mathbb{I}_{32j}]}}\\
&=-6 \cdot a \cdot n_3^*
\end{split}
\end{equation*}

When $i=j$, since $\mathbb{I}_{33j}$ and $\mathbb{I}_{32j}$ are mutually exclusive, $\mathbb{E}[\mathbb{I}_{33i} \cdot \mathbb{I}_{32j}]=0$. Otherwise when $i \neq j$, $\mathbb{I}_{33j}$ and $\mathbb{I}_{32j}$ are independent and $\mathbb{E}[\mathbb{I}_{33i} \cdot \mathbb{I}_{32j}]=0$. Similarly,
\begin{equation*}
\begin{split}
Cov(\sum_{i=1}^{n_3^*}{\mathbb{I}_{33i}}, \sum_{i=1}^{n_1^*}{\mathbb{I}_{32i}})
= - 6 \cdot a \cdot (1-p)^3 \cdot n_3^*
\end{split}
\end{equation*}
and
\begin{equation*}
\begin{split}
Cov(\sum_{i=1}^{n_2^*}{\mathbb{I}_{32i}}, \sum_{i=1}^{n_1^*}{\mathbb{I}_{31i}})
= - 18 \cdot (1-p)^5 \cdot (3p-1) \cdot n_3^*.
\end{split}
\end{equation*}

It then follows that
\begin{equation*}
\begin{split}
&\mathbb{V}ar[LSHE]=n_3^*\cdot\Big(\frac{3\cdot(3p-1)^2\cdot(1-p)^2\cdot(1-3p\cdot(1-p)^2)}{p}\\
&+\frac{(1-6\cdot(1-p)^2\cdot p)^2\cdot(1-p^2\cdot(3-2p))}{p^2\cdot(3-2p)}\\
&-(6\cdot(1-6\cdot(1-p)^2\cdot p)\cdot(3p-1)\cdot(1-p)^2)\\
&+ 9\cdot(1-p)^6+3\cdot(1-p)^3 \Big)\\
&+ n_2^*\cdot\frac{(1-p)}{p}\\
&= n_3^* \cdot \frac{(p-1)^2 \cdot (3p^2-p+1)}{p^2 \cdot (3-2p)} + n_2^* \cdot \frac{(1-p)}{p}.
\end{split}
\end{equation*}

\vspace*{4em}
\subsection{Variance Monotonicity}
\label{mo}
We now prove the monotonicity of the variance of our estimator.
\begin{theorem}
	\label{thm}
	$\mathbb{V}ar[LSHE]$ is monotonically decreasing when $p$ increases in  range $(0,1]$.
\end{theorem}

\begin{proof}
	\begin{lemma}
		\label{lemma}
		First order derivative of $\mathbb{V}ar[LSHE]$ is negative when $p \in (0,1]$.
	\end{lemma}
	\begin{proof}
	Consider
	\begin{align*}
	\dfrac{d(\mathbb{V}[LSHE])}{dp} &= \dfrac{3(2-p)(p-1)(p+1)((p-1)^2+p^2)}{p^3(2p-3)^2} \\ &\quad \cdot n_3^*
	-p^2 \cdot n_2^*.
	\end{align*}
	When $p\in(0,1]$,  $-p^2 <0$. Because $(2-p), (p+1), (p-1)^2+p^2,p^3(2p-3)^2$ are all positive and $(p-1)$ is the only term that is negative, $$\frac{3(2-p)(p-1)(p+1)((p-1)^2+p^2)}{p^3(2p-3)^2} < 0.$$ Thus, $\frac{d(\mathbb{V}ar[LSHE])}{dp} < 0$.
\end{proof}
	By using Lemma \ref{lemma}, we can consequently prove Theorem \ref{thm}. We also note that when $p=1$, $\mathbb{V}ar[LSHE] = 0$.
\end{proof}

\clearpage
\newpage

\bibliographystyle{imsart-nameyear}
%\bibliography{chomp,mybib_merged}
\bibliography{entity}

\end{document}